\documentclass[conference]{IEEEtran}
\IEEEoverridecommandlockouts
\usepackage{cite}
\usepackage{amsmath,amssymb,amsfonts,amsthm}
\usepackage{algorithmic}
\usepackage{float}
\usepackage{tikz}
    \usetikzlibrary{patterns}
\usepackage{pgfplots}
    \pgfplotsset{compat=1.18}
    \pgfplotstableset{col sep=comma}
\usepackage{todonotes}
\usepackage{enumerate,enumitem}
\usepackage{graphicx}
\usepackage{textcomp}
\usepackage{xcolor}
\def\BibTeX{{\rm B\kern-.05em{\sc i\kern-.025em b}\kern-.08em
    T\kern-.1667em\lower.7ex\hbox{E}\kern-.125emX}}
\usepackage[utf8]{inputenc}
\usepackage{booktabs}
\usepackage{multirow}
\usepackage{pifont}
\usepackage{diagbox}
\usepackage{url}

\newtheorem{thm}{Theorem}
\newtheorem{clm}{Claim}
\newtheorem{rem}{Remark}

\newcommand\cA{\mathcal{A}}
\newcommand\cB{\mathcal{B}}
\newcommand\cH{\mathcal{H}}

\newcommand\cX{\mathcal{X}}

\newcommand{\adv}{\ensuremath{\mathsf{adv}}}

\title{Robust Hypothesis Testing with Abstention}
\author{\IEEEauthorblockN{Malhar Managoli\IEEEauthorrefmark{1}, K. R. Sahasranand\IEEEauthorrefmark{2}, Vinod M. Prabhakaran\IEEEauthorrefmark{1}}
\IEEEauthorblockA{\IEEEauthorrefmark{1}School of Tech. and Comp. Sc., Tata Institute of Fundamental Research, Mumbai. \{malhar.managoli,vinodmp\}@tifr.res.in}
\IEEEauthorblockA{\IEEEauthorrefmark{2}Mehta Family School of DS \& AI, Indian Institute of Technology Palakkad.~sahasranand@iitpkd.ac.in}
}

\begin{document}

\maketitle

\begin{abstract} %THIS PAPER IS ELIGIBLE FOR THE STUDENT PAPER AWARD. 
We study the binary hypothesis testing problem where an adversary may potentially corrupt a fraction of the samples. The detector is, however, permitted to abstain from making a decision if (and only if) the adversary is present. We consider a few natural ``contamination models'' and characterize for them the trade-off between the error exponents of the four types of errors --  errors of deciding in favour of the incorrect hypothesis when the adversary is present and errors of abstaining or deciding in favour of the wrong hypothesis when the adversary is absent, under the two hypotheses.
\end{abstract}
All proofs are available in the Appendices.

\section{Introduction}\label{sec:introduction}
In decision making settings such as disease diagnosis and semi-autonomous vehicles, it may be preferable for the system to abstain from making a decision and yield control to an external entity or a human decision maker when the samples are ``out-of-distribution.''  Motivated by this (and modelling ``out-of-distribution'' samples as adversarially generated), we study the following binary hypothesis testing problem between distributions $P_0$ v/s $P_1$. Suppose $n$ independently and identically distributed (i.i.d.) samples $X^n$ are drawn from one of these distributions. Unbeknownst to the detector, the samples {\em might} be intercepted by an adversary who may manipulate some of the samples. The detector observes these potentially contaminated samples. The detector's goal is to determine the underlying distribution ($P_0$ or $P_1$) from which $X^n$ was drawn whenever the samples are uncontaminated. The detector is allowed to abstain from making a decision if the samples are contaminated\footnote{A possible refinement to this model, which we do not consider here, is to require the detector to determine the underlying distribution up to a certain level of corruption by the adversary and permit abstentions only when the corruption exceeds this threshold.}. However, when the samples are uncontaminated, abstaining (in addition to making an incorrect decision) counts as an error; and so does making an incorrect decision when the samples are contaminated.

%This results in four types of errors \dots

There are several strands of related work. Classification with abstention (sometimes also referred to as rejection) dates back to the works of Chow~\cite{chow1957optimum,chow1970optimum} who obtained the optimal Bayes detector and studied the trade-off between error rate and rejection rate. More recently both  passive~\cite{herbei2006classification,bartlett2008classification,yuan2010classification,wegkamp2007lasso,grandvalet2008support,cortes2016learning} and active~\cite{shekhar2021active,zhu2022efficient,denis2024active} versions of the problem have been extensively investigated. Some of the works in the information theory literature which explicitly or otherwise consider abstention include~\cite{grigoryan2011multiple,sason2012moderate,lalitha2016reliability}. Another relevant line of work is that of robust statistics~\cite{tukey1960survey,huber1964robust,huber1965robust,levy2009robust} which addresses modelling errors by considering contamination models (i.e., a family of deviations from the nominal model) and design detectors/estimators with good worst-case performance over them. This problem has attracted significant recent attention~\cite{lai2016agnostic,charikar2017learning,qiao2018learning,diakonikolas2019robust,chen2020learning,jain2020optimal,jain2021robust} where computationally efficient robust estimators have been developed (see~\cite{diakonikolas2023algorithmic} for a recent survey on high-dimensional robust statistics). Closest to our work are a few recent studies~\cite{goldwasser2020identifying,kalai2021towards,goel2023adversarial} which have considered abstention under contamination models. In~\cite{goldwasser2020identifying,kalai2021towards} an algorithm trained on labelled test samples makes predictions on test samples, some of which may be adversarially corrupted; the algorithm is allowed to abstain on corrupted samples and the goal is to achieve small test error and false abstentions. A sequential version of the problem is studied in~\cite{goel2023adversarial}.

We consider three contamination models -- (i) memoryless ingress, where each of the $n$ locations is independently, with probability $\varepsilon$, made available to the adversary to replace; (ii) fixed weight uniform ingress, where $n\varepsilon$ out of $n$ locations are selected uniformly at random and made available for replacement; and (iii) strong contamination, where the adversary may choose $n\varepsilon$ locations to replace. In all cases the adversary makes its choice of what to replace the samples with (and, in strong contamination, which samples to replace) after observing the uncorrupted $X^n$. We characterize the optimal trade-off between the error exponents~\cite{hoeffding1965asymptotically,blahut1974hypothesis,csiszar1971error} for the four types of errors -- worst-case errors of deciding in favour of the incorrect hypothesis when the samples are contaminated and errors of abstaining or deciding in favour of the wrong hypothesis when the samples are uncontaminated. %All missing proofs may be found in the extended version \cite{??}.

\subsubsection*{Notation}
We use $[n]$ to denote the set $\{1,2,\ldots,n\}$. We write $x^n$ to denote an $n$-length vector $(x_1,\ldots,x_n)$. For $z^n\in\{0,1\}^n$, let $x^n|_{z^n}$ denote the vector obtained by ``restricting'' $x^n$ to the indices $i \in [n]$ such that $z_i=1$. Also, let $\overline{z^n}:=1^n-z^n$ be the complement of $z^n$. The Hamming weight of $z^n$ is denoted $\mathrm{wt}(z^n)$. Random variables are denoted by uppercase letters. We say \emph{$X^n$ are i.i.d. $P$} when $X_i, i \in [n]$ are independently and identically distributed with common distribution $P$. We denote this by $X^n\sim P^n$.
We denote the type of $x^n$ by $P_{x^n}$. 
Let $\Delta_n(\cX)$ denote the set of types of length $n$ on $\cX$ and let $\Delta(\cX)$ denote the set of all distributions on $\cX$. We omit the $(\cX)$ when the alphabet is clear from the context. %We say that \emph{$Y^n$ are i.i.d. $P$} when $Y_i, i \in [n]$ are independently and identically distributed with common distribution $P$. 
Let
%\begin{align*}
$B_{\mathrm{KL}}(p,r):=\{q\in\Delta: D(q\|p) \le r\}$ and $B_{\mathrm{TV}}(p,r):=\{q\in\Delta: d_{\mathrm{TV}}(q,p) \le r\}$
%\end{align*}
denote balls of radius $r$ with respect to Kullback-Leibler (KL) divergence and total variation distance, respectively, around the probability distribution $p$. %Let
%\[
%B_{\mathrm{KL}}^{(n)}(p,r) := B_{\mathrm{KL}}(p,r)~\cap~\Delta_n.
%\]
The KL divergence between Bernoulli distributions with parameters $\rho$ and $\varepsilon$ is denoted by $D_2(\rho\|\varepsilon)$.
%We use $B^c$ to denote the complement of the set $B$. 

\section{Problem Statement} \label{sec:problem}
Let $\cX$ be a finite alphabet and $P_0, P_1$ be distinct distributions on $\cX$.
For simplicity, we assume throughout, that $P_0,P_1$ have full support.
The uncorrupted samples $X^n=(X_1,X_2,\ldots,X_n)$ are i.i.d. $P_0$ under hypothesis $\cH_0$ and are i.i.d. $P_1$ under hypothesis $\cH_1$. In the absence of an adversary, the samples %$Y^n$ 
observed by the (potentially randomized) test/detector $\phi_n:\cX^n\to\{0,1,\bot\}$ are identical to $X^n$. When present, the adversary may modify the uncorrupted samples $X^n$, subject to certain constraints, to produce the samples $Y^n\in\cX^n$ observed by the test. We will consider a few different forms of constraints on the adversary. To state the problem generally, let $\cA_{\varepsilon}$ denote the set of all possible (potentially randomized) actions the adversary may take, where the parameter $\varepsilon>0$ denotes the level/budget of corruption. For a given test $\phi_n$, we define four types of errors (we will specify $E^{A}_{1|0}$ and $E^{A}_{0|1}$ for each model more explicitly in the sequel):
\begin{align*}
   E_{1\perp|0}(n)&= \mathbb{P}_{X^n\sim P_0^{n}}\left(\phi_n(X^n) \in \{1,\perp\}\right) \tag*{(no adversary)}\\
   E_{0\perp|1}(n)&= \mathbb{P}_{X^n\sim P_1^{n}}\left(\phi_n(X^n) \in \{0,\perp\}\right) \tag*{(no adversary)}\\
   E^{A}_{1|0}(n)&=\mathbb{P}_{X^n\sim P_0^{n},A}\left(\phi_n(Y^n)=1\right) ~~~~(\text{adversary}~A\in \cA_\varepsilon)\\
   E^{A}_{0|1}(n)&= \mathbb{P}_{X^n\sim P_1^{n},A}\left(\phi_n(Y^n)=0\right) ~~~~(\text{adversary}~A\in \cA_\varepsilon)\\
   E^{\mathsf{adv}}_{1|0}(n)&=\sup_{A\in\cA_{\varepsilon}} E^{A}_{1|0}(n) \tag*{(worst-case error)}\\
   E^{\mathsf{adv}}_{0|1}(n)&=\sup_{A\in\cA_{\varepsilon}} E^{A}_{0|1}(n) \tag*{(worst-case error)}
\end{align*}
i.e., when the adversary is absent, both abstaining and incorrectly detecting the hypothesis count as errors whereas in the presence of an adversary, an error is incurred only when the incorrect hypothesis is detected. 

We say that a positive\footnote{For simplicity, we exclude the ``Chernoff-Stein'' case where some of the exponents may be 0. %The results extend and the regions are continuous as we will discuss in a full version.} 
}
quadruple $\left(\lambda_{1\perp|0}, \lambda_{0\perp|1},\lambda^{\adv}_{1|0}, \lambda^{\adv}_{0|1}\right)$ of exponents is achievable if there is a sequence of tests such that 
\begin{align}
    \lambda_{1\perp|0}&\geq \liminf_{n\to\infty}-\frac{\log E_{1\perp|0}(n)}{n},\label{eq:non-adv|0}\\
    \lambda_{0\perp|1}&\geq \liminf_{n\to\infty}-\frac{\log E_{0\perp|1}(n)}{n},\label{eq:non-adv|1}\\
    \lambda^{\mathsf{adv}}_{1|0}&\geq \liminf_{n\to\infty}-\frac{\log E^{\adv}_{1|0}(n)}{n},\label{eq:adv|0}\\
    \lambda^{\mathsf{adv}}_{0|1}&\geq \liminf_{n\to\infty}-\frac{\log E^{\adv}_{0|1}(n)}{n}.\label{eq:adv|1}
\end{align}
We are interested in characterizing the set $\Lambda$ of all achievable quadruples.

Notice that $E_{1\perp|0}(n)= \mathbb{P}_{X^n\sim P_0^{n}}\left(\phi_n(X^n)\neq 0\right)$ and $E_{0\perp|1}(n)= \mathbb{P}_{X^n\sim P_1^{n}}\left(\phi_n(X^n) \neq 1\right)$ are the (non-adversarial) type 1 and type 2 errors for the binary hypothesis testing problem of $P_0$ v/s $P_1$. Hence, the set of achievable pairs of $(\lambda_{1\perp|0}, \lambda_{0\perp|1})$ is given by the closure\footnote{Strictly speaking, the intersection of the closure with $\mathbb{R}_{>0}^2$ since we exclude the Chernoff-Stein case. This qualification applies to similar statements of the $\Lambda$ regions in the sequel and are omitted.} of the pairs satisfying~\cite[Theorem~16.3]{polyanskiy2024information}
\begin{align}
    B_{\mathrm{KL}}(P_0,\lambda_{1\perp|0})\cap B_{\mathrm{KL}}(P_1,\lambda_{0\perp|1})= \varnothing. \label{eq:non-adv-tradeoff}
\end{align}

%\todo[inline]{Let $z^n \in \{0,1\}^n$. The adversary observes all the samples $X^n$ and replaces the sample $X_k$ with a sample of its choice whenever $z_k = 1, k \in [n]$. We consider the different ways of picking $Z^n$.}
%{\begin{enumerate}
%    \item $Z^n$ are i.i.d. $\mathrm{Ber}(\varepsilon)$.
%    \item $Z^n \sim \mathrm{Unif}\left\{z^n \in \{0,1\}^n: \sum_{i=1}^n z_i = n\varepsilon\right\}$.
%    \item Adversary chooses $z^n$ with Hamming weight at most $n\varepsilon$.
%\end{enumerate}}
%\todo[inline]{The aforementioned ways of choosing the indices to corrupt are instances of what are referred to in the robust statistics literature~\cite{diakonikolas2023algorithmic} as \emph{adaptive contamination models} (and the latter referred to as \emph{strong contamination model}) -- by contrast, under non-adaptive contamination models the adversary is weaker and does not observe the uncorrupted samples.}

\section{Memoryless Ingress Contamination} \label{sec:Bernoulli}
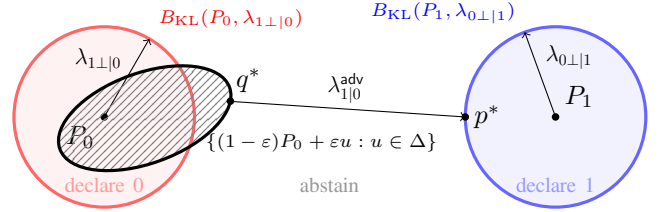
\begin{figure}[H]
    \centering
    \begin{tikzpicture}
        \filldraw[radius=1.2, color=red!60, fill=red!5, very thick](0,0) circle;
        \draw (0,-0.9) node[font = \footnotesize, color=red!40]{declare 0};
        \fill (0,0) circle[radius=0.5mm];
        \draw[->] (0,0) node[below left]{$P_0$}--node[above left, font = \Large, scale=0.6]{$\lambda_{1\perp|0}$} ++({1.2*cos(60)},{1.2*sin(60)})node[color=red!100, font = \scriptsize, above right]{$B_{\mathrm{KL}}(P_0,\lambda_{1\perp|0})$};

        \filldraw[radius=1.2, color=blue!60, fill=blue!5, very thick](6,0) circle;
        \draw (6,-0.9) node[font = \footnotesize, color=blue!40]{declare 1};
        \fill (6,0) circle[radius=0.5mm];
        \draw[->] (6,0) node[above right]{$P_1$} -- node[above right, font = \Large, scale=0.6]{$\lambda_{0\perp|1}$} ++ ({1.2*cos(110)},{1.2*sin(110)}) node[color=blue!100, font = \scriptsize, above left]{$B_{\mathrm{KL}}(P_1,\lambda_{0\perp|1})$} ;

        \draw[x radius=1.2, y radius=0.6, rotate=20, very thick, pattern={north east lines}, pattern color = black!50] (0.5,-0.2) ellipse;
        \draw[->] ({0.5+1.2*cos(20)},{-0.2+1.2*sin(20)}) node[above right]{$q^*$} -- node[above, scale=0.6,font=\Large]{$\lambda^{\adv}_{1|0}$}({6+1.2*cos(180)},{1.2*sin(180)}) node[right]{$p^*$};
        \fill ({0.55+1.2*cos(20)},{-0.2+1.2*sin(20)}) circle[radius=0.5mm];
        \fill ({6+1.2*cos(180)},{1.2*sin(180)}) circle[radius=0.5mm];
        \draw ({.1+1.2*cos(20)},{-0.5+1.2*sin(20)}) node[below right,font=\scriptsize]{$\{(1-\varepsilon)P_0+\varepsilon u:u\in\Delta\}$};
        \draw (3,-.9) node[font=\footnotesize, color = black!40]{abstain};
        
        \end{tikzpicture}
    \caption{Memoryless ingress contamination: The detector declares 0 in the red region, declares 1 in the blue region, and abstains in the white region. The black oval represents $\{(1-\varepsilon)P_0+\varepsilon u:u\in\Delta\}$ which is the set of possible distributions of $Y^n$ when the adversary chooses the replacement samples $\sim u$ i.i.d. and $q^*=(1-\varepsilon)P_0+\varepsilon u^*$, where $(u^*,p^*)$ is the minimiser of (\ref{eq:Bernoulli-1|0}). }
    \label{fig:ber}
\end{figure}

We start with a simple and somewhat weak model for the adversary. Suppose for each $i\in[n]$, whether or not the adversary (when present) will be permitted to replace the sample $X_i$ is determined by a Bernoulli random variable $Z_i\sim \text{Ber}(\varepsilon)$ i.i.d.; furthermore, $Z^n$ is independent of $X^n$. i.e., with probability $\varepsilon$, independent of everything else, each sample is made available to the adversary to replace. We will assume that the adversary may choose the samples to replace $X^n|_{Z^n}$ with after observing $(X^n,Z^n)$. We note here that, while we prove the achievability for this {\em omniscient} adversary, we can show our converse for a weaker {\em oblivious} adversary who is unaware of $X^n,Z^n$ (and specifies an $n$ length vector of potential replacements from which the first $\mathrm{wt}(Z^n)$ elements are used to replace $X^n|_{Z^n}$); thus, the result holds for all models in between these extremes. 
%(Also note that since we are interested in the worst-case error (also because we are interested in the exponent, and arbitrarily large number of uncorrupted samples are available to the adversary), the adversary may be assumed to know the true hypothesis.\todo{confusing? leave out?}).
We note in passing two connections to models from the literature. The oblivious adversary who uses an i.i.d. distribution from the replacement samples is equivalent to Huber's contamination model~\cite{huber1992robust}. If the adversary is allowed only {\em causal} access to $(X^i,Z^i)$ when choosing the replacement $Y_i$, this model can be shown to be related to a special case of adversarial binary hypothesis testing (without abstention) of~\cite{brandao2020adversarial}; however, the results from there do not directly apply to our setting.

Formally, our set of adversaries are all $A:\cX^n\times\{0,1\}^n\to\cX^n$ such that $(A(x^n,z^n))|_{\overline{z^n}}=x^n|_{\overline{z^n}}$.  And we define \[E^{A}_{1|0}(n)=\mathbb{P}_{(X^n,Z^n)\sim P_0^{n}\times(\text{Ber}(\varepsilon))^n}\left(\phi_n(A(X^n,Z^n))=1\right).\]
$E^{A}_{0|1}(n)$ is similarly defined. With these, let $\Lambda^{\mathsf{Ber}(\varepsilon)}$ denote the set of achievable exponent quadruples. Our main result for this section is the following characterization %of $\Lambda^{\mathsf{Ber}(\varepsilon)}$ 
(see Figure~\ref{fig:ber}).

\begin{thm}\label{thm:Bernoulli}
$\Lambda^{\mathsf{Ber}(\varepsilon)}$ is the closure of the set of $(\lambda_{1\perp|0}, \lambda_{0\perp|1},\lambda^{\adv}_{1|0}, \lambda^{\adv}_{0|1})$ such that
    \begin{align}
        &B_{\mathrm{KL}}(P_0,\lambda_{1\perp|0})\cap B_{\mathrm{KL}}(P_1,\lambda_{0\perp|1})= \varnothing, \label{eq:Bernoulli-nonadv}\\
        \lambda^{\adv}_{1|0}&\le \min_{{u\in\Delta,\, p\in B_{\mathrm{KL}}(P_1,\lambda_{0\perp|1})}} D(p\| (1-\varepsilon) P_0 + \varepsilon u), \label{eq:Bernoulli-1|0}\\
        %\min_{\substack{p\in B_{\mathrm{KL}}(P_1,\lambda_{0\perp|1})\\ u\in\Delta}} D(p\| (1-\varepsilon) P_0 + \varepsilon u) \label{eq:Bernoulli-1|0}\\
        \lambda^{\adv}_{0|1}&\le \min_{{u\in\Delta,\, p\in B_{\mathrm{KL}}(P_0,\lambda_{1\perp|0})}} D(p\| (1-\varepsilon) P_1 + \varepsilon u). \label{eq:Bernoulli-0|1}
        %\min_{\substack{p\in B_{\mathrm{KL}}(P_0,\lambda_{1\perp|0})\\ u\in\Delta}} D(p\| (1-\varepsilon) P_1 + \varepsilon u) \label{eq:Bernoulli-0|1}
    \end{align}
\end{thm}

Notice that \eqref{eq:Bernoulli-nonadv} is just the trade-off between the non-adversarial error exponents in \eqref{eq:non-adv-tradeoff}. Furthermore, for a pair $(\lambda_{1\perp|0},\lambda_{0\perp|1})$ which satisfies \eqref{eq:Bernoulli-nonadv}, the result shows that there is no additional tradeoff between the remaining two (adversarial) exponents, something which will become clear from the discussion below. We start with the converse which explains the form the upper bounds in \eqref{eq:Bernoulli-1|0} and \eqref{eq:Bernoulli-0|1} take.

\paragraph*{Converse}To see the converse, it suffices to show \eqref{eq:Bernoulli-1|0} (since \eqref{eq:Bernoulli-0|1} follows by symmetry, and \eqref{eq:Bernoulli-nonadv} follows from \eqref{eq:non-adv-tradeoff}). For this, consider the following (oblivious) adversary $A$ under hypothesis $\cH_0$ who replaces every sample $X_i$, $i\in[n]$ such that $Z_i=1$, by sampling from a distribution $u\in\Delta$ independently\footnote{Notice that we are considering a randomized adversary for showing the converse. This is without loss of generality since the worst-case error over deterministic adversaries is no smaller than that of a randomized adversary.}. Then, the resulting distribution at the detector in the presence of this adversary is the mixture distribution $P_0':=(1-\varepsilon)P_0+\varepsilon u$ i.i.d. Hence, the trade-off between exponents $\lambda^{\adv}_{1|0}$ and $\lambda_{0\bot|1}$ cannot be any better than that of the (non-adversarial) hypothesis testing problem of $P_0'$ v/s $P_1$ (since its type 1 error event is declaring 1 under $P_0'$ and type 2 error is {\em not} declaring 1 under $P_1$). Thus, from~\cite[Theorem 16.3]{polyanskiy2024information}, $$\lambda^{\adv}_{1|0}\le \min_{{p\in B_{\mathrm{KL}}(P_1,\lambda_{0\perp|1})}} D(p\|P_0').$$
%\[\lambda^{\adv}_{1|0}\le \min_{{p\in B_{\mathrm{KL}}(P_1,\lambda_{0\perp|1})}} D(p\| (1-\varepsilon) P_0 + \varepsilon u).\]
We arrive at \eqref{eq:Bernoulli-1|0} by minimizing over the choice of the distribution $u$ used by the adversary (see Figure~\ref{fig:ber}).

\paragraph*{Achievability} To show the achievability, for a pair $(\lambda_{1\perp|0},\lambda_{0\perp|1})$ which satisfy \eqref{eq:Bernoulli-nonadv}, consider the following detector (for $\delta >0$)
\begin{align}
\phi_n(y^n)=
    \begin{cases}
        \text{decide } 0, &D(P_{y^n}\|P_0)\leq  \lambda_{1\perp|0}+\delta\\
        \text{decide } 1, &D(P_{y^n}\|P_1)\leq  \lambda_{0\perp|1}+\delta\\
        \text{abstain } (\bot),&\text{otherwise}
\end{cases}\label{eq:detector]}
\end{align}
Since $(\lambda_{1\perp|0},\lambda_{0\perp|1})$ satisfy \eqref{eq:Bernoulli-nonadv}, the detector is clearly well-defined (for sufficiently small $\delta$). Since
\begin{align*} 
E_{1\bot|0}(n)
&=\mathbb{P}_{X^n\sim P_0^{n}}\left(\phi_n(X^n) \in \{1,\perp\}\right)\\
&=\mathbb{P}_{X^n\sim P_0^{n}}\left(D(P_{y^n}\|P_0)> \lambda_{1\perp|0}+\delta\right)\\
&\leq \exp(-n(\lambda_{1\perp|0}+\delta-o(1))),
\end{align*}
where the inequality follows from Sanov's theorem~\cite[Theorem 11.4.1]{cover1999elements},
we have \eqref{eq:non-adv|0}. Similarly, \eqref{eq:non-adv|1} follows. For the adversarial exponents, we proceed in steps. 
Consider the following claim:
\begin{clm}\label{clm:Bernoulli-achievability}
$\liminf_{n\rightarrow\infty}-\frac{1}{n}\log E^{\adv}_{1|0}(n)$ is at least
    \begin{align}
        %&\lim_{n\rightarrow\infty}-\frac{1}{n}\log E_{1|0}(n)\\ &\qquad\geq
 \min_{\rho\in[0,1]}\min_{\substack{(q,v)\in 
 \cB
 %\Delta^2: D((1-\rho)q + \rho v\|P_1)\le\lambda_{0\bot|1}+\delta}
 }}D_2(\rho\|\varepsilon)+(1-\rho)D(q\|P_0),\label{eq:clm-ber-ach}
    \end{align}
where $\cB=\{(q,v)\in \Delta^2: D((1-\rho)q + \rho v\|P_1)\le\lambda_{0\bot|1}+\delta\}.$
\end{clm}
The interpretation of the above claim is that the exponent for error $E^A_{1|0}(n)$ is dominated by a large deviation event for $Z^n$ in which the weight of $Z^n$ is roughly $\rho n$ (which occurs with probability $\approx\exp(-nD_2(\rho\|\varepsilon))$) {\em and} a large deviation event in which the $n(1-\rho)$ uncorrupted samples have type close to $q$ (which occurs with probability $\approx\exp(-n(1-\rho)D(q\|P_0))$). The worst-case exponent corresponds to minimizing the resulting exponent over $\rho,q$ such that there is a choice $v$ for the distribution of the adversary's samples such that the resulting type $(1-\rho)q+\rho v$ of the input to the detector falls in the region where the detector declares ``1'' (i.e., $D((1-\rho)q + \rho v\|P_1)\le\lambda_{0\bot|1}+\delta$). Note that the adversary's samples do not undergo a large deviation in the dominant event and hence do not contribute to the exponent.

Taking $\delta\to 0$, we conclude
%\footnote{In taking $\delta\to0$, we are using the fact that the set $\Lambda$ is closed by definition.} 
that we achieve the exponent quadruples defined by \eqref{eq:Bernoulli-nonadv}, $\lambda_{1|0}^{\adv}$ upperbounded by the expression in Claim~\ref{clm:Bernoulli-achievability} (with $\delta=0$), and the analogous upperbound for $\lambda_{0|1}^{\adv}$. However, notice that the expression in the claim appears to be different from the one in \eqref{eq:Bernoulli-1|0}. The following claim, proved in Appendix~\ref{app:Bernoulli-equivalence}, 
completes the achievability proof.%\footnote{When taken together with the converse, the two expressions are seen to be equal.}
\begin{clm}\label{clm:Bernoulli-equivalence}
For any $\lambda \in (0,D(P_0\|P_1))$ and $\varepsilon \in (0,1)$, we have
\begin{align*}
 \min_{\rho\in[0,1]}&\min_{\substack{(q,v)\in\Delta^2: D((1-\rho)q + \rho v\|P_1)\le\lambda}}D_2(\rho\|\varepsilon)+(1-\rho)D(q\|P_0)\\ 
 &\qquad\geq
  \min_{{u\in\Delta,\; p\in B_{\mathrm{KL}}(P_1,\lambda)}} D(p\| (1-\varepsilon) P_0 + \varepsilon u).
    \end{align*}
\end{clm}
\begin{rem}
    We note that, taken together with the converse, we may conclude that the above inequality is an equality, a fact which may not be immediately obvious.
\end{rem}

\section{Fixed Weight Uniform Ingress Contamination}\label{sec:constHamming}
\begin{figure}[h]
    \centering
    \begin{tikzpicture}
        \filldraw[radius=1.2, color=red!60, fill=red!5, very thick](0,0) circle;
        \fill (0,0) circle[radius=0.5mm];
        \draw[->] (0,0) node[above left]{$P_0$}--node[above left, font=\Large, scale=.6]{$\lambda_{1\perp|0}$}({1.2*cos(60)},{1.2*sin(60)}) node[color=red!100, font=\scriptsize, above right]{$B_{\mathrm{KL}}(P_0,\lambda_{1\perp|0})$};
        \draw (0,-.7) node[color = red!40]{declare 0};

        \filldraw[radius=1.2, color=blue!60, fill=blue!5, very thick](6,0) circle;
        \fill (6,0) circle[radius=0.5mm];
        \draw[->] (6,0) node[above right]{$P_1$} -- node[above right, font=\Large, scale=.6]{$\lambda_{0\perp|1}$}({6+1.2*cos(100)},{1.2*sin(100)}) node[color=blue!100, font=\scriptsize, above left]{$B_{\mathrm{KL}}(P_1,\lambda_{0\perp|1})$};
        \draw (6.1,-.7) node[color=blue!40]{declare 1};

        \draw[thick, rotate around={30:(4.8,0)}, pattern={north east lines}, pattern color=black!50] (4.8,0) circle[x radius=1.2, y radius=.8];
        \fill (4.8,0) circle[radius=0.5mm];
        \draw (4.8,0) node[above left]{$p^*$};
        \fill ({4.45-sqrt(.5)},0) circle[radius=0.5mm];
        \draw[->] (0,0) -- node[above, font=\Large, scale=.6]{$\qquad\lambda^{\adv}_{1|0}/(1-\varepsilon)$} ({4.45-sqrt(.5)},0) node[above left]{$q^*$};
        \draw ({6.2-0.5*sqrt(0.5)},{-1-0.5*sqrt(0.5)}) node[font=\scriptsize, text width=5cm]{$\big\{q\in\Delta:\exists u\,\mathrm{s.t.}\,\newline \quad(1-\varepsilon)q+\varepsilon u=p^*\big\}$};
        \draw (2.8,-0.7) node[color=black!40]{abstain};
        \end{tikzpicture}
    \caption{Fixed weight uniform ingress contamination: The black oval represents the set of types $q$ of uncorrupted samples ($X^n|_{\overline{Z^n}})$ for which an adversary can induce $P_{Y^n}=p^*$ (for which the detector declares 1). The distribution $p^*=(1-\varepsilon)q^*+\varepsilon u^*$, where $(q^*,u^*)$ is the minimiser of (\ref{eq:Uniform-1|0}). Under $\cH_0$, $P_{X^n|_{\overline{Z^n}}} =q^*$ happens with probability $\approx\exp(-n(1-\varepsilon)D(q^*\|P_0))$.}
    \label{fig:fixed-weight}
\end{figure}
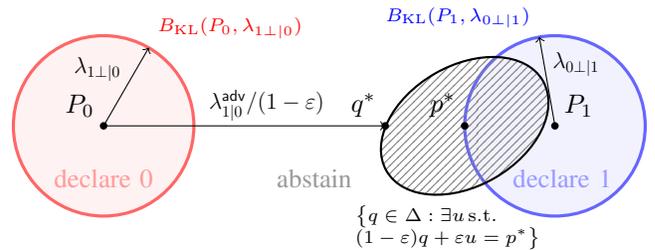
Consider the same setup as above except that the $Z^n\in\{0,1\}^n$ vector, which indicates the locations where the adversary may replace the samples, is now drawn uniformly at random from the set of weight-$\lceil n\varepsilon \rceil$ binary vectors independent of $X^n$. %(we drop $\lceil\,\rceil$ in the sequel). 
As before, we show the achievability for the omniscient adversary (who has access to $(X^n,Z^n)$ when choosing the replacement samples) while our converse can be shown using only an oblivious adversary (who specifies an $n\varepsilon$-length vector of replacement samples which are used to replace $X^n|_{Z^n}$). We denote the tradeoff region for this case by $\Lambda^{\sf{Unif}(\varepsilon)}$.
\begin{thm}\label{thm:Uniform}
$\Lambda^{\sf{Unif}(\varepsilon)}$ is the closure of all quadruples $(\lambda_{1\perp|0}, \lambda_{0\perp|1},\lambda^{\adv}_{1|0}, \lambda^{\adv}_{0|1})$ such that
    \begin{align}
        &B_{\mathrm{KL}}(P_0,\lambda_{1\perp|0})\cap B_{\mathrm{KL}}(P_1,\lambda_{0\perp|1})= \varnothing, \label{eq:Uniform-nonadv}\\
        \lambda^{\adv}_{1|0}&\le \min_{\substack{q,u\in\Delta\ : D((1-\varepsilon)q + \varepsilon u\|P_1)\le\lambda_{0\perp|1}}} (1-\varepsilon)D(q\|P_0), \label{eq:Uniform-1|0}\\
        \lambda^{\adv}_{0|1}&\le \min_{\substack{q,u\in\Delta\ : D((1-\varepsilon)q + \varepsilon u\|P_0)\le\lambda_{1\perp|0}}} (1-\varepsilon) D(q\| P_1). \label{eq:Uniform-0|1}
    \end{align}
\end{thm}

\begin{rem} \label{rem:comparison}
    Note that substituting $\rho=\varepsilon$ into~\eqref{eq:clm-ber-ach} results in~\eqref{eq:Uniform-1|0}. Since~\eqref{eq:clm-ber-ach} is a minimisation over $\rho$, $\lambda_{1|0}^\adv$ is smaller for memoryless ingress than for fixed weight uniform ingress when $\lambda_{0\perp|1}$ is fixed. Similar argument holds for $\lambda_{0|1}^\adv$.
\end{rem}

%\begin{thm}
%    For the robust adversarial hypothesis testing with abstention problem between $P_0$ and $P_1$, when $Z^n\sim\mathrm{Unif}\{z^n\in\{0,1\}^n:\sum_{i=1}^nz_i=n\varepsilon\}$, a tuple $(\lambda_0,\lambda_1,\mu_0,\mu_1)$ belongs to the closure of the set of achievable exponents $\Lambda$ if and only if
    %\begin{align}
     %   B_{\mathrm{KL}}(P_0,\lambda_0)\cap B_{\mathrm{KL}}(P_1,\lambda_1)= \varnothing\\
      %  \mu_0\le(1-\varepsilon) \min_{\substack{q,v\in\Delta:\\ D((1-\varepsilon)q+\varepsilon v\|P_1)\le\lambda_1}} D(q\| P_0)\\
      %  \mu_1\le(1-\varepsilon) \min_{\substack{q,v\in\Delta:\\ D((1-\varepsilon)q+\varepsilon v\|P_0)\le\lambda_0}} D(q\| P_1)
    %\end{align}
%\end{thm}
% \newpage
% ..
% \newpage

The trade-off between the non-adversarial exponents follows from the same arguments as in Theorem~\ref{thm:Bernoulli}. We outline the proof for one of the adversarial exponents, namely, $\lambda^{\adv}_{1|0}$; the other follows similarly.
\paragraph*{Achievability} The detector is the same as before and takes $P_{Y^n}$ as input. We have
\begin{clm}\label{clm:Uniform-achievability}
$\lim_{n\rightarrow\infty}-\frac{1}{n}\log E^{\adv}_{1|0}(n)$ is at least
\begin{align*}
    \min_{\substack{q,u\in\Delta\ : D((1-\varepsilon)q + \varepsilon u\|P_1)\le\lambda_{0\perp|1}+\delta}} (1-\varepsilon) D(q\|P_0).
\end{align*}
\end{clm}
An error occurs when the type of $X^n|_{\overline{Z^n}}$ (denoted $q$) is such that there exists a type $u$ (of length $n\varepsilon$) of the replacement samples such that the resulting type $P_{Y^n} = (1-\varepsilon)q+\varepsilon u$ falls in the decision region for $1$, namely, the set of $n$-length types in $B_{KL}(P_1,\lambda_{1\perp|0}+\delta)$ (see Figure~\ref{fig:fixed-weight}). The exponent is then the large deviation exponent for the occurrence of such a $q$ (when the $n(1-\varepsilon)$ samples are drawn i.i.d. $P_0$), and is given by minimum of $D(q\|P_0)$ over all occurrences of $q$ and all choices of $u$ that result in the error event.
\paragraph*{Converse} We consider an (oblivious) adversary $A$ who restricts itself to drawing the replacement samples i.i.d. from a fixed distribution $u$. %whenever $Z_k = 1$. 
For this restricted adversary, the distribution of $Y^n$ is invariant under permutations under both $\cH_0$ (in the presence of adversary) and $\cH_1$. Hence it suffices to consider (randomized) detectors that take $P_{Y^n}$ as input. We have
\begin{clm}
\label{clm:Uniform-converse}
If  $\liminf_{n\rightarrow\infty} -\frac{1}{n}\log E_{0\perp|1}(n) \ge \lambda$,
    \begin{align*}
        \liminf_{n\rightarrow\infty} &-\frac{1}{n} \log E^{A}_{1|0}(n) \\
        &\le{} (1-\varepsilon) \min_{\substack{q,u\in\Delta: D((1-\varepsilon)q + \varepsilon u\|P_1)\le\lambda}} D(q\|P_0).
    \end{align*}
\end{clm}
Let $g_n(q),q\in\Delta_n$ be the probability that the detector declares $1$ when $P_{Y^n}=q$. Fix $\delta > 0$ and let $g_n$ satisfy the hypothesis in the claim. Then, there exists $N$ such that for all $n > N$, we have
\begin{align}
-\frac{1}{n} \log E_{0\perp|1}(n) \ge \lambda - \delta.
%\label{eq:lamUnif}
\end{align}
Then, by Sanov's theorem, we have for $n > N$,
\begin{align}
\max_{q\in B_{KL}(P_1,{\lambda-2\delta})\cap\Delta_n} \left(1-{g_n}(q)\right)\le \exp({-n\delta/2}).
\label{eq:gBound}
\end{align}
Then, the error probability $E^A_{1|0}(n)$, by~\eqref{eq:gBound}, is at least
\[
(1-\exp({-n\delta/2}))\sum_{q\in B_{KL}(P_1,{\lambda-2\delta})\cap\Delta_n} \mathbb{P}_{\cH_0} [P_{Y^n}=q].
\]
Conditioned on $Z^n = z^n$, the event in the summand is the intersection of $\left\{P_{Y^n|_{z^n}}=q_1\right\}$ and $\left\{P_{Y^n|_{\overline{z^n}}}=q_2\right\}$ where $(1-\varepsilon)q_1+\varepsilon q_2 = q$. The error exponent is then given by Sanov's theorem applied to the corresponding types, namely,
\[
\min_{\substack{q_1,q_2\in\Delta: (1-\varepsilon)q_1 + \varepsilon q_2 \in B_{KL}(P_1,{\lambda-2\delta})}} (1-\varepsilon) D(q_1\|P_0)+\varepsilon D(q_2\|u).
\]
The proof is completed by choosing $q_2 = u$ and letting $\delta \to 0$.
\section{Strong Contamination}\label{sec:advchoose}
\begin{figure}[H]
    \centering
    \begin{tikzpicture}
        \filldraw[radius=1.2, color=red!60, fill=red!5, very thick](0,0) circle;
        \fill (0,0) circle[radius=0.5mm];
        \draw[->] (0,0) node[above left]{$P_0$}--node[above left, font=\Large, scale=.6]{$\lambda_{1\perp|0}$}({1.2*cos(60)},{1.2*sin(60)}) node[color=red!100, font=\scriptsize, above right]{$B_{\mathrm{KL}}(P_0,\lambda_{1\perp|0})$};
        \draw (0,-.8) node[color = red!40]{declare 0};

        \filldraw[radius=1.2, color=blue!60, fill=blue!5, very thick](6,0) circle;
        \fill (6,0) circle[radius=0.5mm];
        \draw[->] (6,0) node[above right]{$P_1$} -- node[above right, font=\Large, scale=.6]{$\lambda_{0\perp|1}$}({6+1.2*cos(110)},{1.2*sin(110)}) node[color=blue!100, font=\scriptsize, above left]{$B_{\mathrm{KL}}(P_1,\lambda_{0\perp|1})$};
        \draw (6,-.8) node[color=blue!40]{declare 1};

        \draw[thick, rotate around={45:(4.8,0)}, pattern={north east lines}, pattern color=black!50] (4.3,-0.5) rectangle (5.3,0.5);
        \fill (4.8,0) circle[radius=0.5mm];
        \draw (4.8,0) node[above left]{$p^*$};
        \fill ({4.8-sqrt(.5)},0) circle[radius=0.5mm];
        \draw[->] (0,0) -- node[above, font=\Large, scale=.6]{$\lambda^{\adv}_{1|0}$} ({4.8-sqrt(.5)},0) node[above left]{$q^*$};
        \draw ({4.9-0.5*sqrt(0.5)},{.1-0.5*sqrt(0.5)}) node[font=\scriptsize, below left]{$B_{\mathrm{TV}}(p^*,\varepsilon)$};
        \draw (2.6,-.8) node[color=black!40]{abstain};
        \end{tikzpicture}
    \caption{Strong contamination: $p^*,q^*$ are minimisers of (\ref{eq:SC|0}). The black rhombus represents $B_{\mathrm{TV}}(p^*,\varepsilon)$, which is the set of types $q$ of the original samples $X^n$, which an adversary can attack to induce $P_{Y^n}=p^*$ (for which the detector declares 1).
    The probability of $P_{X^n}=q^*$ under $\cH_0$ is $\approx\exp(-nD(q^*\|P_0))$.}
    \label{fig:strong}
\end{figure}
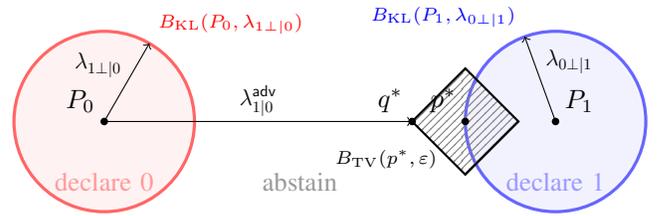

Now consider the setup where the adversary (when present) can choose the attack locations $Z^n$, subject to $\mathrm{wt}(Z^n)\le n\varepsilon$.
Formally, $\cA_{\varepsilon}$ are all $A:\cX^n\rightarrow\cX^n$ such that, for all $x^n\in\cX^n$, $x^n$ and $A(x^n)$ differ in at most $n\varepsilon$ locations.
We denote the trade-off region for this case by $\Lambda^{\mathsf{SC(\varepsilon)}}$.

%\todo[inline]{Malhar, why not $\cA_{\varepsilon}$ are all $A:\cX^n\rightarrow\cX^n$ such that, for all $x^n\in\cX^n$, the number of locations where $x^n$ and $A(x^n)$ differ is at most $n\varepsilon$?}

\begin{thm}
\label{thm:SC}
    $\Lambda^{\mathsf{SC}(\varepsilon)}$ is the closure of all quadruples $(\lambda_{1\perp|0},\lambda_{0\perp|1},\lambda^{\adv}_{1|0},\lambda^{\adv}_{0|1})$ such that
    \begin{align}
        B_{\mathrm{KL}}(P_0,\lambda_{1\perp|0})\cap B_{\mathrm{KL}}(P_1,\lambda_{0\perp|1})= \varnothing \label{eq:SC-nonadv}\\
        \lambda^{\adv}_{1|0} \le\min_{\substack{p,q\in\Delta:\, d_{TV}(q,p)\le\varepsilon,\, D(p\|P_1)\le\lambda_{0\perp|1}}} D(q\| P_0)\label{eq:SC|0}\\
        \lambda^{\adv}_{0|1}\le \min_{\substack{p,q\in\Delta:\, d_{TV}(q,p)\le\varepsilon,\, D(p\|P_0)\le\lambda_{1\perp|0}}} D(q\| P_1)\label{eq:SC|1}
    \end{align}
\end{thm}
Notice again that the trade-off between the non-adversarial exponents is the same as before, and there is no further trade-off between the adversarial exponents.
\paragraph*{Achievability}
Again, we use the same detector as in the previous sections (\ref{eq:detector]}). For the same reasons as before, (\ref{eq:non-adv|0}) and (\ref{eq:non-adv|1}) are satisfied. The following claim gives (\ref{eq:adv|0}); similarly, (\ref{eq:adv|1}) follows.
\begin{clm}
    \label{clm:adv-choice-achievability}
    %For the detector in (\ref{eq:detector]})  
    $\liminf_{n\rightarrow\infty}-\frac{1}{n}\log E_{1|0}^{\adv}(n)$ is at least
    \[ \min_{\substack{p,q\in\Delta:\, d_{TV}(q,p)\le\varepsilon,\, D(p\|P_1)\le\lambda_{0\perp|1}+\delta}} D(q\|P_0).\]
\end{clm}
Note that, in this model, whenever the adversary modifies a sequence $X^n$ to $Y^n$, we must have $d_{\mathrm{TV}}(P_{X^n},P_{Y^n})\le\varepsilon$.
We can then interpret Claim \ref{clm:adv-choice-achievability} as follows (see Figure \ref{fig:strong}).
Suppose there exist distributions $p,q$ such that (a) $p\in B_{\mathrm{KL}}(P_1,\lambda_{0\perp|1}+\delta)$ (i.e. the detector declares 1 when $P_{Y^n}=p$), and (b) $d_{\mathrm{TV}}(p,q)\le \varepsilon$, then whenever $P_{X^n}=q$, the adversary carries out a successful attack, contributing to $E^{\adv}_{1|0}$.
Claim \ref{clm:adv-choice-achievability} says that this (optimised over $p,q$) is the dominant avenue for the adversary's attack.
\paragraph*{Converse}
To show the converse, as before, it suffices to show (\ref{eq:SC|0}).
Consider the following (randomized) adversary $A$:
it first obtains a minimiser $(p^*,q^*)$ of the following problem (for some $\delta>0$).
\begin{align*}
    \min_{\substack{p,q\in\Delta:\, d_{TV}(q,p)\le\varepsilon-\delta,\, D(p\|P_1)\le\lambda_{0\perp|1}}-\delta} D(q\|P_0)
\end{align*}
Then, any input $x^n$ with $P_{x^n}$ close to $q^*$ may be changed to a $y^n$ of type $p_n^*$ with $p^*_n\in\Delta_n$ close to $p^*$ with the $\delta$-back-off ensuring that this can be done within the $n\varepsilon$ budget.
The adversary changes $X^n$ of type $P_{X^n}$ into a $Y^n$ of type $p_n^*$, with the replacement locations and samples chosen at random such that the induced distribution on $Y^n$ is invariant under permutations (see Appendix \ref{sec:SC-proofs} for details).
Therefore, it is enough to consider (randomized) detectors which take $P_{Y^n}$ as input.
Now, similarly to the previous section, $\liminf_{n\rightarrow\infty}-\frac{1}{n}\log E_{0\perp|1}(n) \ge \lambda_{0\perp|1} \implies g_n(p_n^*)\ge 1-\exp({-n\delta/4})$.
At the same time
\begin{align*}
    \mathbb{P}_{\cH_0}[P_{Y^n}=p_n^*]\ge \mathbb{P}_{\cH_0}[P_{X^n}\in B_{\mathrm{KL}}(q^*,\eta)]
\end{align*}
for $\eta>0$ small enough.
The probability on the right-side, by Sanov's theorem, is at least $\exp(-n(D(q^*\|P_0)-o(1)))$.
Therefore,
\begin{align*}
    E^A_{1|0}(n)\ge{}& g_n(p_n^*)\mathbb{P}_{\cH_0}[P_{Y^n}=p_n^*]\\
    \ge {}&\left(1-\exp({-n\delta/4})\right) \exp({-n(D(q^*\|P_0)-o(1)))}.
\end{align*}
Letting $\delta\to 0$ completes the proof.

%\section{Proofs}
%\label{sec:proofs}
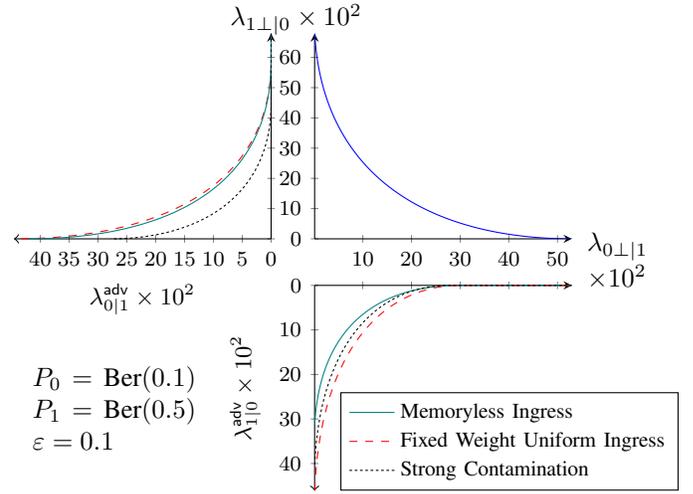
\begin{figure}[h]
    \centering
    \begin{tikzpicture}
        \pgfplotsset{footnotesize}
        %\matrix{
        \begin{axis}[
        name= adv01,
        xlabel= {$\lambda_{0|1}^{\adv}\times10^2$},
        ylabel = {},
        xticklabel = {\pgfmathparse{-100*\tick}$\pgfmathprintnumber{\pgfmathresult}$},
        yticklabel=\empty,
        axis x line = bottom,
        axis y line = right,
        x axis line style={<-},
        ]
        \addplot[red,dashed] table [
        x=La01fw,
        y=L10,
        mark=none,
        ]
        {plot_data.csv};
        \addplot[teal] table [
        x=La01ber,
        y=L10,
        mark=none,
        ]
        {plot_data.csv};
        \addplot[black, dash pattern = on 1pt off 1pt] table [
        x=La01adv,
        y=L10,
        mark=none,
        ]
        {plot_data.csv};
        \end{axis}
        %&
        \begin{axis}[
        xshift=4cm,
        yticklabel = {\pgfmathparse{100*\tick}$\pgfmathprintnumber{\pgfmathresult}$},
        xticklabel = {\pgfmathparse{100*\tick}$\pgfmathprintnumber{\pgfmathresult}$},
        axis x line = bottom,
        axis y line = left,
        ]
        \addplot[blue] table [
        x=L01,
        y=L10,
        mark=none,
        ]
        {plot_data.csv};
        \end{axis}
        %\\
        %text
        %&
        \begin{axis}[
        yshift=-3.35cm, xshift=4cm,
        xlabel= {},
        ylabel = {$\lambda_{1|0}^{\adv}\times10^2$},
        xticklabel=\empty,
        yticklabel = {\pgfmathparse{-100*\tick}$\pgfmathprintnumber{\pgfmathresult}$},
        axis x line = top,
        axis y line = left,
        y axis line style={<-},
        legend style = {at={(.1,0)}, anchor = south west},
        legend cell align = {left},
        ]
        \addplot[teal] table [
        x=L01,
        y=La10ber,
        mark=none,
        ]
        {plot_data.csv};
        \addlegendentry{Memoryless Ingress}
        \addplot[red, dashed] table [
        x=L01,
        y=La10fw,
        mark=none,
        ]
        {plot_data.csv};
        \addlegendentry{Fixed Weight Uniform Ingress}
        \addplot[black, dash pattern = on 1pt off 1pt] table [
        x=L01,
        y=La10adv,
        mark=none,
        ]
        {plot_data.csv};
        \addlegendentry{Strong Contamination}
        \end{axis}
        %\\
        %};
    \draw (7.5,-0.1) node[anchor=west] {$\lambda_{0\bot|1}$};
    \draw (7.5,-0.5) node[anchor=west] {$\times 10^2$};
    \draw (3.75,3.25) node[anchor=north] {$\lambda_{1\bot|0}\times 10^2$};
    \node[text width=2.5cm] at (1.5,-2.25) {
    $P_0=\text{Ber}(0.1)$ 
    $P_1=\text{Ber}(0.5)$
    $\varepsilon=0.1$};
    \end{tikzpicture}
    \caption{Example of the achievability region. Here, $\cX=\{0,1\}$. The top right plot gives the trade-off between the two non-adversarial exponents (note that it is the same for all three adversarial models). For any point $(\lambda_{0\perp|1}, \lambda_{1\perp|0})$ under the blue curve, we can find points $(\lambda_{0|1}^\adv,\lambda_{1\perp|0})$ and $(\lambda_{0\perp|1},\lambda_{1|0}^\adv)$ on the curves in the top left and bottom right plots respectively. These give a quadruple on the boundary of the achievable region.}
    \label{fig:plot1}
\end{figure}

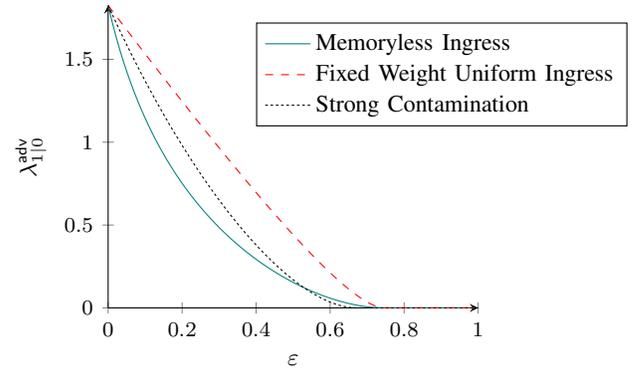
\begin{figure}[h]
    \centering
    \begin{tikzpicture}
        \pgfplotsset{small}
        \begin{axis}
        [
        axis x line = bottom,
        axis y line = left,
        xlabel = {$\varepsilon$},
        ylabel = {$\lambda_{1|0}^{\adv}$},
  %      legend pos = north east,
        legend style={font=\small, at = {(0.4,0.6)}, anchor=south west},
        legend cell align = {left},
        ]
            \addplot[teal] table [
            x=eps,
            y=La10ber,
            mark=none,
            ]
            {plot2_data.csv};
            \addlegendentry{Memoryless Ingress}
            \addplot[red,dashed] table [
            x=eps,
            y=La10fw,
            mark=none,
            ]
            {plot2_data.csv};
            \addlegendentry{Fixed Weight Uniform Ingress}
            \addplot[black, dash pattern = on 1pt off 1pt] table [
            x=eps,
            y=La10adv,
            mark=none,
            ]
            {plot2_data.csv};
            \addlegendentry{Strong Contamination}
        \end{axis}
    \end{tikzpicture}
    \caption{Plot of $\lambda^{\adv}_{1|0}$ against $\varepsilon$ for fixed $\lambda_{0\perp|1}=0.1$. Here, $P_0 = \mathrm{Ber}(0.1)$ and $P_1 = \mathrm{Ber}(0.9)$. The exponent for Strong Contamination is smaller than that for Fixed Weight Uniform Ingress by the problem definition. Remark \ref{rem:comparison} explains why the exponent for Memoryless Ingress is smaller than that for Fixed Weight Uniform Ingress. Strong Contamination is not comparable with Memoryless Ingress.}
    \label{fig:plot2}
    \vspace{-0.3cm}
\end{figure}

Figure~\ref{fig:plot1} shows the optimal trade-off between the exponents and Figure~\ref{fig:plot2} shows the dependence of $\lambda^{\adv}_{1|0}$ on the corruption level $\varepsilon$, both for the binary alphabet. 
\appendices

%\clearpage
\bibliographystyle{IEEEtran}
%\bibliography{biblio}
%\nocite{*}

%\end{document}

\clearpage
\section*{Additional notation for the appendices}
We use the shorthand notation $$B_{\mathrm{KL}}^{(n)}(p,\lambda):=B_{\mathrm{KL}}(p,\lambda) \cap \Delta_n.$$

All logs and exponents are base $2$.
\section{Proof of Claim~\ref{clm:Bernoulli-equivalence}}\label{app:Bernoulli-equivalence}
We have
\begin{align}
&\min_{p\in B_{\mathrm{KL}}(P_1,\lambda)}\min_{u\in\Delta}D\left(p\|(1-\varepsilon)P_0+\varepsilon u \right)\nonumber\\
    &={}\min_{p\in B_{\mathrm{KL}}(P_1,\lambda)}\min_{\rho\in[0,1]}\min_{\substack{q_1,q_2\in\Delta:\\(1-\rho)q_1+\rho q_2=p}}\nonumber\\
    &\quad \min_{u\in\Delta(\mathcal{X})}D\left((1-\rho)q_1+\rho q_2\|(1-\varepsilon)P_0+\varepsilon u \right).
    \label{eq:KLterm}
\end{align}
\allowdisplaybreaks
Define two joint distributions for $Z,Y$ as follows:
\begin{align*}
    &P^{(0)}_{Z}={} \mathrm{Ber}(\varepsilon),\quad P^{(0)}_{Y|Z=0} ={} P_0, \quad P^{(0)}_{Y|Z=1} ={} u;\\
    &P^{(1)}_{Z}={}\mathrm{Ber}(\rho), \quad P^{(1)}_{Y|Z=0}={}q_1, \quad P^{(1)}_{Y|Z=1}={}q_2.
\end{align*}
Then, the KL divergence term on the right-side of~\eqref{eq:KLterm} is %given by
\begin{align*}
    D&\left((1-\rho)q_1+\rho q_2\middle\|(1-\varepsilon)P_0+\varepsilon u \right)\\
    &=D\left(P^{(1)}_Y\big\| P^{(0)}_Y \right)\\
    &\stackrel{\text{(a)}}{\le} D\left(P^{(1)}_{ZY}\big\| P^{(0)}_{ZY} \right)\\
    &\stackrel{\text{(b)}}{=} D\left(P^{(1)}_Z\middle\| P^{(0)}_Z \right) + D\left(P^{(1)}_{Y|Z}\middle\| P^{(0)}_{Y|Z}\middle| P^{(1)}_Z\right)\\
    &=D(\rho\|\varepsilon) + (1-\rho) D(q_1\| P_0) + \rho D(q_2\| v),
\end{align*}
where (a) is by the data processing inequality and (b) follows by chain rule. Choosing $q_2 = v$ yields
\begin{align*}
&\min_{p\in B_{\mathrm{KL}}(P_1,\lambda)}\min_{u\in\Delta}D\left(p\|(1-\varepsilon)P_0+\varepsilon u \right)\\
&\le \min_{p\in B_{\mathrm{KL}}(P_1,\lambda)}\min_{\rho\in[0,1]}\min_{\substack{q_1,v\in\Delta:\\(1-\rho)q_1+\rho v=p}} D(\rho\big\|\varepsilon) + (1-\rho) D(q_1\big\| P_0)\\
&= \min_{\rho\in[0,1]}\min_{\substack{q,v\in\Delta:\,\mathrm{s.t.}\\ D((1-\rho)q +\rho v\|P_1)\le\lambda}}D(\rho\|\varepsilon)+(1-\rho)D(q\|P_0).
\end{align*}
\section{Proofs for Memoryless Ingress Contamination}
\label{sec:Bernoulli-proofs}
%Here, Nature chooses ${Z^n}$ according to $\mathrm{Ber}(\varepsilon)$ i.i.d.
%\subsection*{Achievability}
% \begin{lemma}
%     There is a sequence of tests such that $\lim_{n\rightarrow\infty}-\frac{1}{n}\log(E_{0\perp|1}(n))= \lambda$ for $\lambda\in D(P_0\|P_1)$ and
%     \begin{align*}
%         &\lim_{n\rightarrow\infty}-\frac{1}{n}\log E_{1|0}(n)^{(a)}\\
%         &= \min_{p\in B_{\mathrm{KL}}(P_1,\lambda)}\min_{u\in\Delta}D\left(p\|(1-\varepsilon)P_0+\varepsilon u \right)
%     \end{align*}
% \end{lemma}
\begin{proof}[Proof of Claim \ref{clm:Bernoulli-achievability}]
Recall that the detector declares 1 if and only if $P_{Y^n}\in B_{KL}(P_1,\lambda+\delta)$.
%by~\cite[Theorem~$11.2.1$]{cover1999elements}, there is some $n_\lambda$ such that for all $n>n_\lambda$, we have
%\[
%\mathbb{P}_{Y^n\sim P_1^{\otimes n}}\left[D\left(P_{Y^n}\|P_1\right)\le\lambda+\delta\right]>1-2^{-n(\lambda+\delta/2)}.
%\]
%This ensures that $E_{0\perp|1}(n)<2^{-n\lambda}$ for all $n>n_\lambda$.

Throughout this proof, we use the shorthand notation 
\[
D_{\varepsilon}(P,Q\|R) := D((1-\varepsilon)P + \varepsilon Q \| R).
\]
We fix an adversary $A \in \cA_{\varepsilon}$ and calculate the probability of declaring 1 under $\cH_0$:
\begin{align*}
    E_{1|0}^{A}(n)
    &={}\mathbb{P}[D(P_{Y^n}\|P_1)\le\lambda+\delta]\\
    &={} \sum_{k=0}^{n}\sum_{\substack{z^n\in\{0,1\}^n:\\
    \mathrm{wt}(z^n)=k}} \mathbb{P}[Z^n = z^n] \cdot \\
    &\qquad \qquad\quad\mathbb{P}\left[ D(P_{Y^n}\|P_1)\le\lambda+\delta\middle| Z^n=z^n\right].
\end{align*}
The second term in the summand is given by
\begin{align*}
    &\mathbb{P}\left[ D(P_{Y^n}\|P_1)\le\lambda+\delta\middle| Z^n=z^n\right]\\
    %%%%%%%%%%%%%%%%
    ={}& \mathbb{P}\left[ D_{\frac{k}{n}}\left(P_{Y^n|_{\overline{z^n}}}, P_{Y^n|_{z^n}} \Big\| P_1\right) \le\delta\middle| Z^n=z^n\right]\\
    %%%%%%%%%%%%%%%%
    \le{}& \mathbb{P}\left[\exists v\in \Delta_{k}: D_{\frac{k}{n}}\left(P_{Y^n|_{\overline{z^n}}}, v \Big\| P_1\right) \le\lambda+\delta \middle| Z^n=z^n\right]\\
    %%%%%%%%%%%%%%%%
    ={}& \sum_{\substack{q\in\Delta_{n-k}:\exists v\in\Delta_{k}\, \mathrm{s.t.}\\
    D_{\frac{k}{n}}(q,v\|P_1)\le\lambda+\delta}} \mathbb{P}_{\cH_0}[P_{X^{n-k}}=q].\\
    %%%%%%%%%%%%%%%%
    \le{}& (n+1)^{|\cX|}\exp\left\{ -\left(n-k\right)\left( \min_{\substack{q\in\Delta_{n-k}:\exists v\in\Delta_{k}\, \mathrm{s.t.}\\
    D_{\frac{k}{n}}(q,v\|P_1)\le\lambda+\delta}} D(q\|P_0)\right)\right\}\\
    %%%%%%%%%%%%%%%%
    =:{}& e_k 
\end{align*}
where the last inequality follows from Sanov's theorem~\cite[Theorem $11.4.1$]{cover1999elements}. Therefore, we have
\begin{align*}
    E_{1|0}^{A}(n) \le{}& \sum_{k=0}^{n} \mathbb{P}[\mathrm{wt}(Z^n)=k] \cdot e_k\\
    \le{}& \sum_{k=0}^{n} \exp\left(-nD\left(\frac{k}{n}\Big\|\varepsilon\right)+2\log(k+1)\right) \cdot e_k
\end{align*}
by Sanov's theorem. Then,    
\begin{align*}
&E_{1|0}^{A}(n) \\
    \le{}& (n+1) \cdot \max_{k\in\{0\}\cup[n]}\exp\left(-nD\left(\frac{k}{n}\Big\|\varepsilon\right)+2\log(k+1)\right) \cdot e_k
\end{align*}
The right-side is evaluated to be $\exp(-nK_n$) where
\begin{align*}
    K_n ={}& \min_{k\in[n]}\min_{\substack{q\in\Delta_{n-k}:\exists v\in\Delta_{k}\, \mathrm{s.t.}\\
    D_{\frac{k}{n}}(q,v\|P_1)\le\lambda+\delta}}D\left(\frac{k}{n}\Big\|\varepsilon\right) + \left(1-\frac{k}{n}\right)D(q\|P_0)\\
    &\qquad\qquad \qquad \qquad\quad -\frac{2 \log(k+1)}{n} - \frac{|\cX|\log(n+1)}{n}.
\end{align*}
The above bound on $E_{1|0}^{A}(n)$ is independent of $A \in \cA_{\varepsilon}$ and hence applicable for $\sup_{A \in \cA_{\varepsilon}}E^{A}_{1|0}(n)$. That is, we have
\begin{align*}
&\liminf_{n \to \infty} -\frac{1}{n} \log E^{\mathsf{adv}}_{1|0}(n)\\
\ge{}&\min_{\rho\in[0,1]}\min_{\substack{q,v\in\Delta\,\mathrm{s.t.}\\ D_\rho(q,v\|P_1)\le\lambda+\delta}}D(\rho\|\varepsilon)+(1-\rho)D(q\|P_0).
\end{align*}
\end{proof}
%The result follows from  claim \ref{clm:Bernoulli-equivalence} and the claim below.
\begin{clm}
    \begin{align*}
    &\lim_{\delta\rightarrow0} \min_{\rho\in[0,1]} \min_{\substack{q,v\in\Delta\,\mathrm{s.t.}\\ D_\rho(q, v\|P_1)\le\lambda+\delta}} D(\rho\|\varepsilon) + (1-\rho) D(q\|P_0)\\
    ={}&\min_{\rho\in[0,1]}\min_{\substack{q,v\in\Delta:\,\mathrm{s.t.}\\ D_\rho(q, v\|P_1)\le\lambda}}D(\rho\|\varepsilon)+(1-\rho)D(q\|P_0).
    \end{align*}
\end{clm}
\begin{proof}
For each $\delta$, the above minimisation problem has a compact domain, and the objective function is convex as a function of $(\rho,q,v)$.
Therefore, this problem has a minimizer $(\rho_\delta, q_\delta, v_\delta)$.
Take any sequence $\{ \delta_m \}_{m\in\mathbb{N}}$ which converges to 0 as $m\rightarrow\infty$.
Consider any corresponding sequence of minimizers $\{(\rho_{\delta_m}, q_{\delta_m}, v_{\delta_m}) \}_{m\in\mathbb{N}}$.
Take any limit point of this sequence $(\rho_{\infty},q_\infty,v_\infty)$. This must satisfy $D(q_\infty,v_\infty\|P_0) \le\lambda$ (note that $P_0$ has full support).
That is, $(\rho_{\infty},q_\infty,v_\infty)$ is a feasible point for the minimization problem on the right-side.
Therefore
\begin{align*}
    &\min_{\rho\in[0,1]} \min_{\substack{q,v\in\Delta:\,\mathrm{s.t.}\\ D_\rho(q, v\|P_1)\le\lambda}}D(\rho\|\varepsilon)+(1-\rho)D(q\|P_0)\\
    \le{}& D(\rho_\infty\|\varepsilon)+(1-\rho_\infty) D(q_\infty\|P_0)
\end{align*}
Since this is true for an arbitrary sequence $\{ \delta_m \}_{m\in\mathbb{N}}$ and any convergent subsequence of $\{ (\rho_{\delta_m}, q_{\delta_m}, v_{\delta_m}) \}_{m\in\mathbb{N}}$, this shows
\begin{align*}
    &\limsup_{\delta\rightarrow0} \min_{\rho\in[0,1]} \min_{\substack{q,v\in\Delta\,\mathrm{s.t.}\\ D_\rho(q, v\|P_1)\le\lambda+\delta}} D(\rho\|\varepsilon) + (1-\rho) D(q\|P_0)\\
    ={}&\min_{\rho\in[0,1]}\min_{\substack{q,v\in\Delta:\,\mathrm{s.t.}\\ D_\rho(q, v\|P_1)\le\lambda}}D(\rho\|\varepsilon)+(1-\rho)D(q\|P_0).
\end{align*}
Since the left-side is an increasing function of $\delta$, this completes the proof.    
\end{proof}   
\section{Proofs for Fixed Weight Uniform Ingress Contamination}
\label{sec:fixed-weight-proofs}
\subsection*{Achievability}
% \begin{lemma}
% There is a sequence of tests such that $\liminf_{n\rightarrow\infty}-\frac{1}{n}\log(E_{0\perp|1}(n))=\lambda$ for $\lambda\in (0,D(P_0\|P_1))$ and
% \begin{align}
%     &\liminf_{n\rightarrow\infty}-\frac{1}{n}\log(E_{1|0}^{\adv}(n))\\
%     &\ge (1-\varepsilon) \min_{\substack{q,v\in\Delta\,\mathrm{s.t.}\\ D((1-\varepsilon)q + \varepsilon v\|P_1)\le\lambda}} D(q\|P_0)\label{eq:fw ach}
% \end{align}
% \end{lemma}
\begin{proof}[Proof of Claim \ref{clm:Uniform-achievability}] 
%  Fix $\delta>0$.
% By~\cite[Theorem~$11.2.1$]{cover1999elements}, there is some $n_\lambda$ such that for all $n>n_\lambda$, we have
% \[
% \mathbb{P}_{Y^n\sim P_1^{\otimes n}}\left[D\left(P_{Y^n}\|P_1\right)>\lambda+\delta\right]\le2^{-n(\lambda+\delta/2)}.
% \]
% Recall that under $\cH_1$, $Y^n$ are i.i.d. $P_1$. So, we use the following test:  declare 1 iff $P_{Y^n}\in B_{\mathrm{KL}}^{(n)}(P_1,\lambda+\delta)$.
% This ensures that $E_{0\perp|1}(n)<2^{-n\lambda}$ for all $n>n_\lambda$.

Throughout this proof, we use the shorthand notation 
\[
D_{\varepsilon}(P,Q\|R) := D((1-\varepsilon)P + \varepsilon Q \| R).
\]
For any adversary $A$ and for any $z^n$ of Hamming weight $n\varepsilon$, under $\mathcal{H}_0^{A}$, we have
\begin{align*}
    &\mathbb{P} \left[D\left(P_{Y^n}\| P_1\right) \le\lambda+\delta|Z^n=z^n\right]\\
    %%%%%%%%%%%%%%%%
    ={}& \mathbb{P} \left[D_\varepsilon\left(P_{Y^n|_{\overline{z^n}}}, P_{Y^n|_{z^n}}\middle\| P_1\right)\le\lambda+\delta \middle| Z^n=z^n\right]\\
    %%%%%%%%%%%%%%%
    \le{}& \mathbb{P} \left[\exists v\in\Delta_{n\varepsilon}:D_\varepsilon\left(P_{Y^n|_{\overline{z^n}}}, v\middle\|P_1\right)\le\lambda+\delta \middle| Z^n=z^n\right]\\
    %%%%%%%%%%%%%%%
    ={}& \sum_{\substack{q\in\Delta_{n(1-\varepsilon)}, v\in\Delta_{n\varepsilon}\, \mathrm{s.t.}\\ D_\varepsilon\left(q,v\| P_1\right)\le\lambda+\delta}} \mathbb{P}_{X^{n(1-\varepsilon)}\sim P_0^{\otimes n(1-\varepsilon)}}[P_{X^{n(1-\varepsilon)}}=q]\\
    %%%%%%%%%%%%%%%
    \le{}& \exp\left(-n(1-\varepsilon)\left( \min_{\substack{q,v\in\Delta \, \mathrm{s.t.} \\ D_\varepsilon \left(q, v\|P_1\right) \le\lambda+\delta}} D(q\|P_0)-o(1)\right)\right).
\end{align*}
where the last step follows by Sanov's theorem~\cite[Theorem $11.4.1$]{cover1999elements}. Therefore, averaging over all $z^n$ of Hamming weight $n\varepsilon$, we have for $n > n_\lambda$,
\begin{align*}
&E^{A}_{1|0}(n)\le\mathbb{P}_{\cH_0}[D(P_{Y^n}\|P_1)\le\lambda+\delta]\\
    \le{}&\exp\left(-n(1-\varepsilon)\left(\min_{\substack{q,v \in\Delta\,\mathrm{s.t.}\\ D_\varepsilon \left(q,v\| P_1\right) \le\lambda+\delta}} D(q\|P_0)-o(1)\right)\right)
\end{align*}
whereby
\begin{align*}
    \liminf_{n\rightarrow\infty}-\frac{1}{n}\log E^{A}_{1|0}(n)\ge{}&(1-\varepsilon) \min_{\substack{q,v\in\Delta\, \mathrm{s.t.}\\ D_\varepsilon\left(q, v\| P_1\right)\le\lambda+\delta}}D\left(q\|P_0\right).
\end{align*}
Since this is true for any adversary $A$, the proof is complete.
\end{proof}
\begin{clm}
    \[\lim_{\delta\rightarrow0^+} \min_{\substack{q,v\in\Delta\, \mathrm{s.t.}\\ D_\varepsilon\left(q, v\| P_1\right) \le\lambda+\delta}} D\left(q\|P_0\right) = \min_{\substack{q,v\in\Delta\, \mathrm{s.t.}\\ D_\varepsilon(q, v\| P_1) \le \lambda}} D\left(q\|P_0\right)\]
\end{clm}
\begin{proof}
    For each $\delta$, the above minimisation problem has a compact non-empty domain, and the objective function is convex as a function of $(q,v)$.
    Therefore, this problem has a minimizer $(q_\delta, v_\delta)$.
    Take any sequence $\{ \delta_m \}_{m\in\mathbb{N}}$ which converges to 0 as $m\rightarrow\infty$.
    Consider any corresponding sequence of minimizers $\{(q_{\delta_m}, v_{\delta_m}) \}_{m\in\mathbb{N}}$.
    Take any limit point of this sequence, $(q_{\infty},v_\infty)$. This must satisfy $D_\varepsilon(q_{\infty},v_\infty
    \|P_1) \le\lambda$.
    That is, $(q_{\infty},v_\infty)$ is a feasible point for the minimization problem on the right-side.
    Therefore
    \begin{align*}
        \min_{\substack{q,v\in\Delta\, \mathrm{s.t.}\\ D_\varepsilon\left(q, v\| P_1\right) \le\lambda}} D(q\|P_0)\le D(q_\infty\|P_0)
    \end{align*}
    Since this is true for an arbitrary sequence $\{ \delta_m \}_{m\in\mathbb{N}}$ and any convergent subsequence of $\{ (p_{\delta_m}, q_{\delta_m}) \}_{m\in\mathbb{N}}$, this shows
\begin{align*}
    &\limsup_{\delta\rightarrow0^+} \min_{\substack{q,v\in\Delta\, \mathrm{s.t.}\\ D_\varepsilon\left(q, v\| P_1\right) \le\lambda+\delta}} D(q\|P_0)\\
    \ge{}&\min_{\substack{q,v\in\Delta\, \mathrm{s.t.}\\ D_\varepsilon\left(q, v\| P_1\right) \le\lambda}} D(q\|P_0).
\end{align*}
Since the left-side is an increasing function of $\delta$, this completes the proof.
\end{proof}
% Therefore,
% \begin{align*}
%     \liminf_{n\rightarrow\infty}-\frac{1}{n}\log E^{\adv}_{1|0}(n)\ge{}&(1-\varepsilon) \min_{\substack{q,v\in\Delta\, \mathrm{s.t.}\\ D_\varepsilon(q, v\| P_1) \le \lambda}}D\left(q\|P_0\right).
% \end{align*}
\subsection*{Converse}
% \begin{lemma}
% If  $\liminf_{n\rightarrow\infty} -\frac{1}{n}\log\left(E_{0\perp|1}(n)\right)\ge \lambda$,
%     \begin{align*}
%         &\liminf_{n\rightarrow\infty} -\frac{1}{n} \log\left(E^{\adv}_{1|0}(n) \right)\\
%         \le{}& (1-\varepsilon) \min_{\substack{q,v\in\Delta\,\mathrm{s.t.}\\ D((1-\varepsilon)q + \varepsilon v\|P_1)\le\lambda}} D(q\|P_0)
%     \end{align*}
% \end{lemma}
\begin{proof}[Proof of Claim \ref{clm:Uniform-converse}]
Consider an adversary $A$ who restricts themselves to drawing i.i.d. samples $W_k$ from a fixed distribution $v$ whenever $Z_k = 1$. For this restricted adversary, the distribution of $Y^n$ is invariant under permutations under both $\cH^{A}_0$ and $\cH_1$. Hence it suffices to consider (randomized) detectors that take $P_{Y^n}$ as input.

Let a sequence of tests be described by a function $$g_n:\Delta_n(\cX) \rightarrow[0,1]$$ where $g_n(q),q\in\Delta_n$ is the probability that the test declares 1 when $P_{Y^n}=q$. Also, let $\overline{g}_n$ denote $1-g_n$.
Consider any sequence of tests satisfying 
$$\liminf_{n\rightarrow\infty} -\frac{1}{n} \log E_{0\perp|1}(n)=\lambda .$$
Fix a $\delta>0$. There exists an $n_\lambda$ such that for all $n>n_\lambda$,
\begin{align*}
    & 2^{-n(\lambda-\delta)} \\
    \ge{}&\sum_{q\in\Delta_n}\overline{g}(q) \mathbb{P}_{\cH_1} [P_{Y^n}=q]\\
    \ge{}&\sum_{q\in B^{(n)}_{KL}(P_1,{\lambda-2\delta})} \overline{g_n}(q) \mathbb{P}_{\cH_1}[P_{Y^n}=q]\\
    \ge{}&\min_{q\in B^{(n)}_{KL}(P_1,{\lambda-2\delta})}\mathbb{P}_{\cH_1}[P_{Y^n}=q] \sum_{q\in B^{(n)}_{KL}(P_1,{\lambda-2\delta})} \overline{g_n}(q)\\
    \ge{}&\min_{q\in B^{(n)}_{KL}(P_1,{\lambda-2\delta})}\mathbb{P}_{\cH_1}[P_{Y^n}=q] \max_{q\in B^{(n)}_{KL}(P_1,{\lambda-2\delta})} \overline{g_n}(q)\\
    ={}& 2^{-n(\lambda-2\delta+o(1))} \max_{q\in B^{(n)}_{KL}(P_1,{\lambda-2\delta})} \overline{g_n}(q).
\end{align*}
where the last line follows from Sanov's theorem\cite[Theorem 11.4.1]{cover1999elements}.
Therefore,
\begin{align*}
    \max_{q\in B^{(n)}_{KL}(P_1,{\lambda-2\delta})} \overline{g}(q) \le{}& 2^{-n(\delta-o(1))}.
\end{align*}
Thus, for $n$ large enough, we have
\[
\max_{q\in B^{(n)}_{KL}(P_1,{\lambda-2\delta})} \overline{g}(q)\le 2^{-n\delta/2}.
\]
Therefore, for any such $n$, we have
\begin{align*}
    &E^{A}_{1|0}(n)\ge{} \sum_{q\in\Delta_n} g(q) \mathbb{P}_{\cH_0^A} [P_{Y^n}=q]\\
    \ge{}& \sum_{q\in B^{(n)}_{KL}(P_1,{\lambda-2\delta})} g_n(q) \mathbb{P}_{\cH_0^A} [P_{Y^n}=q]\\
    \ge{}& \min_{q\in B^{(n)}_{KL}(P_1,{\lambda-2\delta})}g_n(q) \sum_{q\in B^{(n)}_{KL}(P_1,{\lambda-2\delta})} \mathbb{P}_{\cH_0^A} [P_{Y^n}=q]\\
    \ge{}& \left(1-2^{-n\delta/2}\right) \sum_{q\in B^{(n)}_{KL}(P_1,{\lambda-2\delta})} \mathbb{P}_{\cH_0^A} [P_{Y^n}=q].
\end{align*}
For any $z^n$ of Hamming weight $n\varepsilon$, consider
\begin{align*}
    &\sum_{q\in B^{(n)}_{KL}(P_1,{\lambda-2\delta})} \mathbb{P}_{\cH_0^A} [P_{Y^n}=q|Z^n=z^n]\\
    %%%%%%%%%%%%%
    ={}& \sum_{\substack{q_1\in\Delta_{n(1-\varepsilon)},\\ q_2\in\Delta_{n\varepsilon}:\\ (1-\varepsilon)q_1+\varepsilon q_2=q\\ q\in B^{(n)}_{KL}(P_1,{\lambda-2\delta})}} \mathbb{P}_{\cH_0^A} [P_{Y^n|_{z^n}}=q_1\mathrm{\ and\ } P_{Y^n|_{\overline{z^n}}}=q_2]\\
    %%%%%%%%%%%%%%%%%
    ={}& \sum_{\substack{q_1\in\Delta_{n(1-\varepsilon)},\\ q_2\in\Delta_{n\varepsilon}:\\ (1-\varepsilon)q_1+\varepsilon q_2=q\\ q\in B^{(n)}_{KL}(P_1,{\lambda-2\delta})}} \mathbb{P}_{X^{n(1-\varepsilon)}\sim P_0^{ n(1-\varepsilon)}} \left[P_{X^{n(1-\varepsilon)}}=q_1\right]\\
    &\qquad\qquad\qquad\qquad \cdot\mathbb{P}_{W^{n\varepsilon}\sim v^{n\varepsilon}} \left[ P_{W^{n\varepsilon}}=q_2\right]\\
\end{align*}
\begin{align*}
    \ge{}& \sum_{\substack{q_1\in\Delta_{n(1-\varepsilon)},\\ q_2\in\Delta_{n\varepsilon}:\\ (1-\varepsilon)q_1+\varepsilon q_2=q\\ q\in B^{(n)}_{KL}(P_1,{\lambda-2\delta})}} 2^{-n(1-\varepsilon)(D(q_1\|P_0)+o(1))} 2^{-n\varepsilon (D(q_2\|v)+o(1))}\\
    %%%%%%%%%%%%%%%%%%%%%%%%%%
    \ge{}& \max_{\substack{q_1\in\Delta_{n(1-\varepsilon)}\\ q_2\in\Delta_{n\varepsilon}\\ (1-\varepsilon)q_1 + \varepsilon q_2 = q\\ q\in B^{(n)}_{KL}(P_1,{\lambda-2\delta})}} 2^{-n\left((1-\varepsilon) D(q_1\|P_0)+\varepsilon D(q_2\|v)+o(1)\right)}
\end{align*}
where the last step follows from the fact that $|\Delta_n|,|\Delta_{n\varepsilon}|, |\Delta_{n(1-\varepsilon)}|$ are only poly($n$) and so contribute only $o(1)$ to the exponent.

Note that this bound is independent of $z^n$. Furthermore, since $\bigcup_{n\ge1}\Delta_n$ is dense in $\Delta$, for any $q,q_1,q_2\in\Delta$, we can find a sequence $q^{(n)}\in\Delta_n, q_1^{(n)}\in\Delta_{n(1-\varepsilon)}, q_2^{(n)}\in\Delta_{n\varepsilon}$ which approaches it. In particular the sequence of minimizers of $$\min_{\substack{q_1\in\Delta_{n(1-\varepsilon)}\\ q_2\in\Delta_{n\varepsilon}\\ (1-\varepsilon)q_1 + \varepsilon q_2 = q\\ q\in B^{(n)}_{KL}(P_1,{\lambda-2\delta})}} (1-\varepsilon) D(q_1\|P_0)+\varepsilon D(q_2\|v)$$ will approach the minimizer of $$\min_{\substack{q,q_1,q_2\in\Delta:\\ (1-\varepsilon)q_1 + \varepsilon q_2 = q\\ q\in B_{KL}(P_1,{\lambda-2\delta})}} (1-\varepsilon) D(q_1\|P_0)+\varepsilon D(q_2\|v)$$

Therefore we have
\begin{align*}
    E_{1|0}^{A}\ge{}& \left(1-2^{n\delta/2}\right)\\
    &\cdot\exp\left(-n\left( \min_{\substack{q,q_1,q_2\in\Delta:\\ (1-\varepsilon)q_1 + \varepsilon q_2 = q\\ q\in\Delta_n\cap B_{KL}(P_1,{\lambda-2\delta})}} (1-\varepsilon) D(q_1\|P_0)\right.\right.\\
    &\Bigg.\Bigg.+\varepsilon D(q_2\|v)+o(1)\Bigg)\Bigg)
\end{align*}
By choosing $v=q_2$, we get
\begin{align*}
    \liminf_{n\rightarrow\infty}-\frac{1}{n}\log E^{\adv}_{1|0}(n)\le{}& \liminf_{n\rightarrow\infty}-\frac{1}{n}\log E^{A}_{1|0}(n)\\
    \le{}&(1-\varepsilon) \min_{\substack{q_1,v\in\Delta\\
    q\in B_{KL}(P_1,{\lambda-2\delta}):\\ (1-\varepsilon) q_1+\varepsilon v=q}} D(q_1\|P_0)
\end{align*}
%The following claim completes the proof.
\end{proof}
\begin{clm}
    \[ \lim_{\delta\rightarrow0} \min_{\substack{q_1,v\in\Delta\\
    q\in B_{KL}(P_1,{\lambda-2\delta}):\\ (1-\varepsilon) q_1+\varepsilon v=q}} D(q_1\|P_0) = \min_{\substack{q_1,v\in\Delta\\
    q\in B_{KL}(P_1,{\lambda}):\\ (1-\varepsilon) q_1+\varepsilon v=q}} D(q_1\|P_0)\]
\end{clm}
\begin{proof}
    The minimisation problem on the right-side has a compact, non-empty domain, and $D(q_1\|P_0)$ is a continuous function of $(q_1,v,q)$.
    Therefore it has a minimiser $(q_1^*,v^*,q^*)$.
    For any $\eta>0$, we can find a $\delta>0$ small enough, so that $B_{\mathrm{KL}}(P_1,\lambda-2\delta)\cap B_{\mathrm{TV}}(q^*,\eta)\ne\varnothing$ (Note that we have assumed $P_0,P_1$ to have full support).
    Pick a $q'\in B_{\mathrm{KL}}(P_1,\lambda-2\delta)\cap B_{\mathrm{TV}}(q^*,\eta)$.
    There exist $q_1',v'\in\Delta$ such that $(1-\varepsilon)q_1'+\varepsilon v'=q'$ and $d_{\mathrm{TV}}(q_1',q_1^*) \le\frac{\eta}{1-\varepsilon}$.
    Since $d_{\mathrm{TV}}(q_1',q_1^*)$ can be made arbitrarily small, we can make $D(q_1'\|P_0)-D(q_1^*\|P_0)\le\epsilon$ for any $\epsilon>0$.
    Since $(q_1',v',q')$ is a feasible point for the minimisation problem on the left-side, for $\delta$ small enough,
    \begin{align*}
        \min_{\substack{q_1,v\in\Delta\\
        q\in B_{KL}(P_1,{\lambda-2\delta}):\\ (1-\varepsilon) q_1+\varepsilon v=q}} D(q_1\|P_0) \ge \min_{\substack{q_1,v\in\Delta\\
        q\in B_{KL}(P_1,{\lambda}):\\ (1-\varepsilon) q_1+\varepsilon v=q}} D(q_1\|P_0)+\epsilon
    \end{align*}
    Since the expression on the left-side is an increasing function of $\delta$, this completes the proof.
\end{proof}

\section{Proofs for Strong Contamination}
\label{sec:SC-proofs}
\subsection*{Achievability}
\begin{proof}[Proof of Claim \ref{clm:adv-choice-achievability}]
% The test is the same as in the previous section. Fix $\delta>0$. Declare 1 if and only if $P_{Y^n}\in B_{KL}(P_1,\lambda+\delta)$. Again by~\cite[Theorem~$11.2.1$]{cover1999elements}, there is some $n_\lambda$ such that for all $n>n_\lambda$, we have
% \[
% \mathbb{P}_{Y^n\sim P_1^{\otimes n}}\left[D\left(P_{Y^n}\|P_1\right)\le\lambda+\delta\right]>1-2^{-n(\lambda+\delta/2)}.
% \]
% This ensures that $E_{0\perp|1}(n)<2^{-n\lambda}$ for all $n>n_\lambda$.
For any fixed adversary $A$,
    \begin{align*}
        E_{1|0}^A={}& \mathbb{P}_{\cH_0}[D\left(P_{Y^n}\|P_1\right)\le\delta]\\
        \le{}& \mathbb{P}_{\cH_0} \left[ \exists p\in\Delta_n: d_{\mathrm{TV}}(p,P_{X^n})\le\varepsilon,\,D(p\|P_1)\le\lambda+\delta\right]\\
        ={}& \sum_{\substack{q\in\Delta_n:\\ \exists p\in \Delta_n\, \mathrm{s.t.}\\ d_{\mathrm{TV}}(p,q)\le\varepsilon\\ D(p\|P_1)\le\lambda+\delta}} \mathbb{P}_{X^n\sim P_0^{\otimes n}} [P_{X^n}=q]\\
        ={}& \exp\left(-n\left(\min_{\substack{p,q\in\Delta_n\, \mathrm{s.t.}\\ d_{\mathrm{TV}}(p,q)\le\varepsilon\\ D(p\|P_1)\le\lambda+\delta}} D(q\|P_0)-o(1)\right)\right)
    \end{align*}
    Here, the last inequality follows from Sanov's theorem \cite[Theorem 11.4.1]{cover1999elements}.    
    % Let $\mathcal{D}(\delta):=\{(p,q)\in \Delta\times\Delta: d_{\mathrm{TV}}(p,q)\le\varepsilon, D(p\|P_1)\le\lambda+\delta, D(q\|P_0)<\infty\}$.
    % If $\mathcal{D}(\delta)\ne\varnothing$, then $\mathcal{D}(2\delta) \cap\Delta_n\ne\varnothing$ for $n$ large enough.
    Since this is true for any adversary, we have
    \[-\frac{1}{n}\log E^{\adv}_{1|0} \ge \min_{\substack{p,q\in\Delta\, \mathrm{s.t.}\\ d_{\mathrm{TV}}(p,q)\le\varepsilon\\ D(p\|P_1)\le\lambda+2\delta}} D(q\|P_0)-o(1)\]
%    The following claim completes the proof
\end{proof}
    \begin{clm}
        \[ \lim_{\delta\rightarrow0} \min_{\substack{p,q\in\Delta\, \mathrm{s.t.}\\ d_{\mathrm{TV}}(p,q)\le\varepsilon\\ D(p\|P_1)\le\lambda+\delta}} D(q\|P_0) = \min_{\substack{p,q\in\Delta\, \mathrm{s.t.}\\ d_{\mathrm{TV}}(p,q)\le\varepsilon\\ D(p\|P_1)\le\lambda}} D(q\|P_0)\]
        \label{clm:adv_ach}
    \end{clm}
    \begin{proof}
%        If for some $\delta>0$, the minimisation problem on the left-side has an empty domain, then that on the right-side also has an empty domain - both sides are $\infty$. 
        For each $\delta$, the minimisation problem on the left-side has a compact, non-empty domain, and the objective function is convex as a function of $(p,q)$.
        Therefore, this problem has a minimizer $(p_\delta, q_\delta)$.
        Take any sequence $\{ \delta_m \}_{m\in\mathbb{N}}$ which converges to 0 as $m\rightarrow\infty$.
        Consider any corresponding sequence of minimizers $\{(p_{\delta_m}, q_{\delta_m}) \}_{m\in\mathbb{N}}$.
        Take any limit point of this sequence $(p_{\infty},q_\infty)$. This must satisfy $D(p_{\infty}\|P_1)\le\lambda$.
        %(this is because all $p_\delta$'s here are absolutely continuous with respect to $P_1$).
        That is, $(p_{\infty},q_\infty)$ is a feasible point for the minimization problem on the right-side.
        Therefore
        \begin{align*}
            &\min_{\substack{q\in\Delta_n:\\ \exists p\in \Delta_n\, \mathrm{s.t.}\\ d_{\mathrm{TV}}(p,q)\le\varepsilon,\\ D(p\|P_1)\le\lambda}} D(q\|P_0) \le{} D(q_\infty\|P_0).
        \end{align*}
        Since this is true for an arbitrary sequence $\{ \delta_m \}_{m\in\mathbb{N}}$ and any convergent subsequence of $\{ (p_{\delta_m}, q_{\delta_m}) \}_{m\in\mathbb{N}}$, this shows
        \begin{align*}
            &\limsup_{\delta\rightarrow0^+} \min_{\substack{q\in\Delta_n:\\ \exists p\in \Delta_n\, \mathrm{s.t.}\\ d_{\mathrm{TV}}(p,q)\le\varepsilon,\\ D(p\|P_1)\le\lambda+\delta}} D(q\|P_0)
            \ge{} \min_{\substack{q\in\Delta_n:\\ \exists p\in \Delta_n\, \mathrm{s.t.}\\ d_{\mathrm{TV}}(p,q)\le\varepsilon,\\ D(p\|P_1)\le\lambda}} D(q\|P_0).
        \end{align*}
        Since the left-side is an increasing function of $\delta$, this completes the proof.    
    \end{proof}
\subsection*{Converse}
% \begin{lemma}
%     %If $\lambda_{0\perp|1}\ge \lambda>\min_{\substack{p\in\Delta: p\left(\overline{\mathrm{supp}(P_0)}\right)\le\varepsilon}} D(p\|P_1)$,
%     If  $\liminf_{n\rightarrow\infty} -\frac{1}{n}\log\left(E_{0\perp|1}(n)\right)\ge \lambda$,
%     \begin{align*}
%         &\liminf_{n\rightarrow\infty} -\frac{1}{n} \log\left(E^{(a)}_{1|0}(n) \right)\\
%         \le{}& \min_{p,q\in \Delta:d_{\mathrm{TV}}(p,q)\le\varepsilon,\,D(p,P_1)\le\lambda}D\left(q\|P_0\right)
%     \end{align*}
% \end{lemma}
\begin{proof}[Details of proof of Converse of Theorem \ref{thm:SC}]
    We design the adversary's attack in steps.

    \emph{Step 1:}
    Fix some $\delta>0$ with $\delta<\varepsilon,\delta<\lambda$ (Later we will take the limit $\delta\rightarrow0$).
    Define \[p^*,q^*:=\arg\min_{\substack{p,q\in\Delta:\\ d_{\mathrm{TV}}(p,q) \le\varepsilon-\delta,\\ D(p\|P_1)\le\lambda-\delta}}D(q\|P_0)\]
    
    \emph{Step 2:}
    Choose $\eta>0$ small enough so that $B_{\mathrm{KL}}(q^*,\eta)\subset B_{\mathrm{TV}}(q^*,\delta/2)$.
    This is possible due to Pinsker's inequality. The adversary attacks iff $P_{X^n}\in B_{\mathrm{TV}}(q^*,\delta/2) $.
    
    \emph{Step 3:}
    Set $p_n^*:=\arg\min_{p_n\in\Delta_n} D(p_n\|p^*)$.
    Then, $p_n^*$ is the target of the adversary's attack.
    We need to ensure that, the adversary will have enough budget to change $P_{X^n}$ into $p_n^*$.
    Since $D(p_n^*\|p^*)$ tends to 0 as $n\rightarrow\infty$, using Pinsker's inequality, there exists an $N$ such that for $n>N$, $d_{\mathrm{TV}}(p_n^*,p^*) \le \delta/2$.
    Therefore, $\forall n>N$, 
    \begin{align*}
    d_{\mathrm{TV}}(P_{X^n},p_n^*)\le{}& d_{\mathrm{TV}}(P_{X^n},q^*) + d_{\mathrm{TV}}(q^*,p^*) + d_{\mathrm{TV}}(p^*,p_n^*)\\
    ={}& \delta/2 + (\varepsilon-\delta) + \delta/2 = \varepsilon
    \end{align*}
    Therefore, the adversary has enough budget to change $P_{X^n}$ to $p_n^*$ for $n>N$.

    \emph{Step 4:}
    We know that the adversary can convert any sequence of a type $q\in B_{\mathrm{KL}}(q^*,\eta)$ to some sequence of type $p_n^*$. We now describe how to (randomly) pick a sequence in the type $p_n^*$ so that, the resulting distribution on $Y^n$ is invariant under permutation. This way, $P_{Y^n}$ will be a sufficient statistic for deciding whether or not to declare 1 based on $Y^n$.

    To change $P_{X^n}$ to $p_n^*$:
    Initialise $Y^n=X^n, Z^n=0^n$.
    \begin{enumerate}
        \item Pick an $x\in\cX$ with $P_{Y^n}(x)>p_n^*(x)$, pick a $y\in\cX$ with $P_{Y^n}(y)<p_n^*(y)$.
        \item Pick $i$ uniformly at random from $\{i:Y_i=x\}$. Set $Z_i=1$ and $Y_i=y$.
    \end{enumerate}
    Repeat the steps 1) and 2) until the $P_{Y^n}$ is of type $p_n^*$. Since $d_{\mathrm{TV}}(P_{X^n},p_n^*) \le\varepsilon$, the resulting $Z^n$ will have Hamming weight at most $n\varepsilon$.

    Suppose the (randomized) test is given by $g:\Delta_n\rightarrow[0,1]$, so that probability of declaring 1 when $P_{Y^n}=q$ is $g(q)$ for all $q\in \Delta_n$.

    Since $\lim_{n\rightarrow\infty}p_n^*= p^*$ and $P_1$ has full support, $\lim_{n\rightarrow\infty}D(p_n^*\|P_1) = D(p^*\|P_1)=\lambda-\delta$.
    Thus, for $n$ large enough, $D(p_n^*\|P_1)\le \lambda-\delta/2$.
    Therefore, for $n$ large enough, $\mathbb{P}_{\cH_1}[P_{Y^n}=p_n^*]\ge 2^{-n(\lambda-\delta/4)}$.
    
    Since $-\frac{1}{n} \liminf_{n\rightarrow\infty} \log E_{0\perp|1}(n)\ge \lambda$, for $n$ large enough $E_{0\perp|1}(n)\le 2^{-n(\lambda-\delta/8)}$.
    This implies
    \begin{align*}
        2^{-n(\lambda-\delta/8)}\ge{}& E_{0\perp|1}(n)\\
        ={}& \sum_{p\in\Delta_n}\overline{g}(p) \mathbb{P}_{\cH_1}[P_{Y^n}=p]\\
        \ge{}& \overline{g}(p^*) \mathbb{P}_{\cH_1}[P_{Y^n}=p^*]\\
        \ge{}& (1-g(p^*))2^{-n(\lambda-\delta/4)}\\
        \implies g(p_n^*)\ge{}& 1-2^{-n(\delta/8)}
    \end{align*}

    So, for $n$ large enough, whenever $P_{X^n}\in B_{\mathrm{KL}}(q^*,\eta)$, the adversary carries out an attack successfully and the algorithm makes an error with a constant probability.
    Thus,
    \begin{align*}
        E_{1|0}^{A}\ge{}& \mathbb{P}_{\cH_0} [P_{X^n}\in B_{\mathrm{KL}}(q^*,\eta)]\left(1-2^{-n(\delta/8)}\right)\\
        ={}& \sum_{q\in B_{\mathrm{KL}}^{(n)}(q^*,\eta)} \mathbb{P}_{\cH_0}[P_{X^n}=q]\left(1-2^{-n(\delta/8)}\right)\\
        \ge{}& \max_{q\in B_{\mathrm{KL}}^{(n)}(q^*,\eta)} \mathbb{P}_{\cH_0}[P_{X^n}=q]\left(1-2^{-n(\delta/8)}\right)\\
        \ge{}& \max_{q\in B_{\mathrm{KL}}^{(n)}(q^*,\eta)} 2^{-n(D(q\|P_0)+o(1))}\\
        \ge{}& \exp\left(-n\left(\min_{q\in B_{\mathrm{KL}}(q^*,\eta)} D(q\|P_0)+o(1)\right)\right)
    \end{align*}
    Note that $\bigcup_{n\ge1} B_{\mathrm{KL}}^{(n)}(q^*,\eta)$ is a dense subset of $B_{\mathrm{KL}}(q^*,\eta)$. Also, $P_0$ has full support, so $D(q\|P_0)$ is continuous on this domain. This is why, in the last step we can minimise over the entire $B_{\mathrm{KL}}(q^*,\eta)$.

   The following claim completes the proof.
   \begin{clm}
        \[ \lim_{\delta\rightarrow0}\min_{\substack{p,q\in\Delta:\\ d_{\mathrm{TV}}(p,q) \le\varepsilon-\delta\\ D(p\|P_1)\le\lambda-\delta}}D(q\|P_0)= \min_{\substack{p,q\in\Delta:\\ d_{\mathrm{TV}}(p,q) \le\varepsilon\\ D(p\|P_1)\le\lambda}}D(q\|P_0)\]
   \end{clm}
    \begin{proof}
        The minimisation problem on the right-side has a compact, non-empty domain, and $D(q\|P_0)$ is a continuous function of $(p,q)$.
        Therefore it has a minimiser $(p^*,q^*)$.
        For any $\eta>0$, we can find a $\delta>0$ small enough, so that $B_{\mathrm{KL}}(P_1,\lambda-\delta)\cap B_{\mathrm{TV}}(p^*,\eta)\ne\varnothing$.
        Pick a $p'\in B_{\mathrm{KL}}(P_1,\lambda-\delta)\cap B_{\mathrm{TV}}(p^*,\eta)$.
        There exists $q'\in\Delta$ such that $d_{\mathrm{TV}}(p',q')\le\varepsilon-\delta$ and $d_{\mathrm{TV}}(q',q^*) \le\eta+\delta$.
        Since $d_{\mathrm{TV}}(q',q^*)$ can be made arbitrarily small, we can make $D(q'\|P_0)-D(q^*\|P_0)\le\epsilon$ for any $\epsilon>0$ (due to Pinsker's inequality).
        Since $(p',q')$ is a feasible point for the minimisation problem on the left-side, for $\delta$ small enough,
        \begin{align*}
            \min_{\substack{p,q\in\Delta:\\ d_{\mathrm{TV}}(p,q)\le\varepsilon-\delta,\\
            p\in B_{KL}(P_1,\lambda-\delta)}} D(q_1\|P_0) \ge \min_{\substack{p,q\in\Delta:\\ d_{\mathrm{TV}}(p,q)\le\varepsilon,\\
            p\in B_{KL}(P_1,\lambda)}} D(q_1\|P_0)+\epsilon
        \end{align*}
        Since the left-side is an increasing function of $\delta$, this completes the proof.
    \end{proof}
\end{proof}

\end{document}